%% file: main.tex
\documentclass[preprint,numbers,nocopyrightspace]{sigplanconf}

\usepackage{alltt}
\usepackage{bcprules}
\usepackage{slashbox}
\usepackage{proof}
\usepackage{latexsym}
\usepackage{amsmath,amsfonts,amssymb,amsthm}
\usepackage{nccmath}
\usepackage{color}
\usepackage{url}
\usepackage{bussproofs}
\usepackage{comment}
\usepackage{lscape}
\usepackage{xspace}
\usepackage{enumitem}
\usepackage{wrapfig}
\usepackage{centernot}
\usepackage{stmaryrd}
\usepackage{pifont}
\usepackage{threeparttable}
\newcommand{\cmark}{\ding{51}}
\newcommand{\xmark}{\ding{55}}

\newcommand{\full}[2]{#1}

\input{macro}

\begin{document}

\setlength{\pdfpageheight}{\paperheight}
\setlength{\pdfpagewidth}{\paperwidth}

\title{Automating Induction for Solving Horn Clauses}

\authorinfo{Hiroshi Unno \and Sho Torii}
           {University of Tsukuba}
           {\{uhiro,sho\}@logic.cs.tsukuba.ac.jp}

\maketitle

\begin{abstract}
\input{abstract}
\end{abstract}

\section{Introduction}
\label{sec:intro}
\input{intro}

\section{Overview of Induction-Based Horn Constraint Solving Method}
\label{sec:overview}
\input{overview}

\section{Horn Constraint Solving Problems}
\label{sec:horn}
\input{horn}

\section{Induction-based Horn Constraint Solving Method}
\label{sec:induction}
\input{induction}

\section{Implementation and Preliminary Experiments}
\label{sec:exp}
\input{exp}

\section{Related Work}
\label{sec:related}
\input{related}

\section{Conclusion and Future Work}
\label{sec:conc}
\input{conc}

\newpage
\appendix

\section{Proof of Lemma~\ref{lem:sound}}
\label{soundness}

\input{soundness}

\end{document}

%% file: macro.tex
\newtheorem{theorem}{Theorem}
\newtheorem{lemma}{Lemma}
\newtheorem{corollary}{Corollary}

\newtheorem{proposition}{Proposition}
\newtheorem{principle}{Principle}

\newtheorem{definition}{Definition}
\newtheorem{example}{Example}

\renewcommand{\restriction}{\mathord{\upharpoonright}}

\newcommand{\coleq}{\hspace{1pt}\colon\!\!\colon\!\!\!\!=}

\newcommand{\h}{{\cal H}}

\newcommand{\D}{{\cal D}}

\newcommand{\hc}{\mathit{hc}}

\newcommand{\leastmodel}{\mu F_{\h}}
\newcommand{\leastmodeld}[1]{\mu F_{#1}}

\newcommand{\bm}[1]{\mbox{\boldmath $#1$}}

\newcommand{\Amult}{A}
\newcommand{\Jmult}{J_{\mathit{mult}}}

\newcommand{\Dmult}{D_{\mathit{mult}}}

\newcommand{\hmult}{\h_{\mathit{mult}}}
\newcommand{\dcmult}{\D_{\mathit{mult}}}

\newcommand{\rhomult}{\rho_{\mathit{mult}}}
\newcommand{\thetamult}{\theta_{\mathit{mult}}}

\newcommand{\Gammamult}{\Gamma_{\mathit{mult}}}

\newcommand{\defof}[1]{\mathit{def}(#1)}
\newcommand{\goalof}[1]{\mathit{goal}(#1)}
\newcommand{\fvs}[1]{\mathit{fvs}(#1)}
\newcommand{\pvs}[1]{\mathit{pvs}(#1)}

\newcommand{\dom}[1]{\mathrm{dom}(#1)}
\newcommand{\arity}[1]{\mathit{ar}(#1)}
\newcommand{\set}[1]{\left\{ #1 \right\}}
\newcommand{\relmiddle}[1]{\mathrel{}\middle#1\mathrel{}}
\newcommand{\seq}[1]{\widetilde{#1}}
\newcommand{\sembrack}[1]{\left\llbracket#1\right\rrbracket}

\newcommand{\natset}{\mathbb{N}}
\newcommand{\intset}{\mathbb{Z}}

\newcommand{\TINT}{\mathtt{int}}
\newcommand{\TREAL}{\mathtt{real}}
\newcommand{\TBOOL}{\mathtt{bool}}
\newcommand{\TLIST}{\mathtt{list}}
\newcommand{\TUNIT}{\mathtt{unit}}

\newcommand{\main}{\mathtt{main}}
\newcommand{\mult}{\mathtt{mult}}
\newcommand{\multacc}{\mathtt{mult\_acc}}

\newcommand{\MULT}{\mathtt{mult}}

\newcommand{\SUM}{\mathtt{sum}}

\newcommand{\SUMACC}{\mathtt{sum\_acc}}
\newcommand{\SUMFROMDOWN}{\mathtt{sum\_down}}
\newcommand{\SUMFROMUP}{\mathtt{sum\_up}}

\newcommand{\APPLY}{\mathtt{apply}}

\newcommand{\REPEAT}{\mathtt{repeat}}

\newcommand{\ADD}{\mathtt{add}}

\newcommand{\NONINTER}{\mathtt{noninter}}
\newcommand{\MCC}{\mathtt{mc91}}
\newcommand{\ACK}{\mathtt{ack}}
\newcommand{\MULTC}{\mathtt{mult\_Ccode}}
\newcommand{\FMOD}{\mathtt{flip\_mod}}
\newcommand{\FIND}{\mathtt{find}}
\newcommand{\FINDOPT}{\mathtt{find\_opt}}
\newcommand{\MEM}{\mathtt{mem}}
\newcommand{\SUML}{\mathtt{sum\_list}}
\newcommand{\FOLDL}{\mathtt{fold\_left}}
\newcommand{\FOLDR}{\mathtt{fold\_right}}

\newcommand{\LINT}{\mathcal{T}_{\intset}}
\newcommand{\QFLIA}{QF\_LIA\xspace}
\newcommand{\QFLRA}{QF\_LRA\xspace}

\newcommand{\ASSERT}{\mathtt{assert}}

\newcommand{\FALSE}{\mathtt{false}}
\newcommand{\NONE}{\mathtt{None}}
\newcommand{\SOME}{\mathtt{Some}}
\newcommand{\NOTFOUND}{\mathtt{Not\_Found}}
\newcommand{\NIL}{\mathtt{Nil}}
\newcommand{\CONS}{\mathtt{Cons}}

\newcommand{\notreds}{\longrightarrow^* \hspace{-1.8em}/\quad}

\newcommand\rulesp{\vspace*{2ex}}
\newcommand{\fresh}{\mathrm{fresh}}

\newcommand{\foreach}{\mathrm{for~each}}

\newcommand{\mk}[2]{{#1}^{#2}}
\newcommand{\mkk}[3]{{#1}_{#3}^{#2}}

\newcommand{\In}[2]{\left\llbracket#1 {\bm \in} #2\right\rrbracket}
\newcommand{\Sub}[2]{\left\llbracket#1 {\bm \subseteq} #2\right\rrbracket}

\newcommand{\sem}[3]{\sembrack{#1, #2}^{#3}}

%% file: abstract.tex
Verification problems of programs written in various paradigms (such
as imperative, logic, concurrent, functional, and object-oriented
ones) can be reduced to problems of solving Horn clause constraints on
predicate variables that represent unknown inductive invariants.
This paper presents a novel Horn constraint solving method based on
inductive theorem proving: the method reduces Horn constraint solving
to validity checking of first-order formulas with inductively defined
predicates, which are then checked by induction on the derivation of
the predicates.  To automate inductive proofs, we introduce a novel
proof system tailored to Horn constraint solving and use an SMT solver
to discharge proof obligations arising in the proof search.
The main advantage of the proposed method is that it can
verify \emph{relational specifications} across programs in various
paradigms where multiple function calls need to be analyzed
simultaneously.  The class of specifications includes practically
important ones such as functional equivalence, associativity,
commutativity, distributivity, monotonicity, idempotency, and
non-interference.
Furthermore, our novel combination of Horn clause constraints with
inductive theorem proving enables us to naturally and automatically
axiomatize recursive functions that are possibly non-terminating,
non-deterministic, higher-order, exception-raising, and over
non-inductively defined data types.
We have implemented a relational verification tool for the OCaml
functional language based on the proposed method and obtained
promising results in preliminary experiments.

%% file: intro.tex
Verification problems of programs written in various paradigms,
including imperative~\cite{Gurfinkel2015}, logic,
concurrent~\cite{Gupta2011a},
functional~\cite{Unno2008,Rondon2008,Unno2009,Vazou2014}, and
object-oriented~\cite{Kahsai2016} ones, can be reduced to problems of
solving Horn clause constraints on predicate variables that represent
unknown inductive invariants.  A given program is guaranteed to
satisfy its specification if the Horn constraints generated from the
program have a solution (see \cite{Grebenshchikov2012} for an overview
of the approach).

This paper presents a novel Horn constraint solving method based on
inductive theorem proving: the method reduces Horn constraint solving
to validity checking of first-order formulas with inductively defined
predicates, which are then checked by induction on the derivation of
the predicates.  The main technical challenge here is how to automate
inductive proofs.  To this end, we propose an inductive proof system
tailored for Horn constraint solving and an SMT-based technique to
automate proof search in the system.

Compared to previous Horn constraint solving
methods~\cite{Unno2009,Terauchi2010,Hoder2011,Gupta2011,Grebenshchikov2012,Rummer2013,McMillan2013,Unno2015}
based on Craig interpolation~\cite{Craig1957a,McMillan2005}, abstract
interpretation~\cite{Cousot1977}, and PDR~\cite{Bradley2011}, the
proposed method has two major advantages:
\begin{enumerate}
\item It can verify \emph{relational specifications} where multiple function
calls need to be analyzed simultaneously.
As shown in Sections~\ref{sec:hornex} and \ref{sec:exp}, the class of
specifications includes practically important ones such as functional
equivalence, associativity, commutativity, distributivity,
monotonicity, idempotency, and non-interference.
\item It can solve Horn clause constraints over whatever background theories
supported by the underlying SMT solver.  Example constraints in
Section~\ref{sec:hornex} are over the theories of nonlinear
arithmetics and algebraic data structures, which have not been
supported by available Horn constraint solvers to our knowledge.
\end{enumerate}

To show the usefulness of our approach, we have implemented a
relational verification tool for the OCaml functional language based
on the proposed method and obtained promising results in preliminary
experiments.

For an example of the reduction from (relational) verification to Horn
constraint solving, consider the following functional program $\Dmult$
(in OCaml syntax).\footnote{Our work also applies to programs that
require a path-sensitive analysis of intricate control flows caused by
non-termination, non-determinism, higher-order functions, and
exceptions but, for illustration purposes, we use this simple program
as a running example.}
\begin{alltt}
let rec mult x y =
  if y=0 then 0 else x + mult x (y-1)
let rec mult_acc x y a =
  if y=0 then a else mult_acc x (y-1) (a+x)
let main x y a =
  assert (mult x y + a = mult_acc x y a)
\end{alltt}
Here, the function $\mult$ takes two integer arguments $\verb|x|$,
$\verb|y|$ and recursively computes $\verb|x| \times \verb|y|$ (note
that $\mult$ never terminates if $\verb|y|<0$).  $\multacc$ is a
tail-recursive version of $\mult$ with an accumulator $\verb|a|$.  The
function $\main$ contains an assertion with the condition
\verb|mult x y + a = mult_acc x y a|, which represents a relational
specification, namely, the functional equivalence of $\mult$ and
$\multacc$.
Our verification problem here is whether for any integers $x$, $y$,
and $a$, the evaluation of $\verb|main|\ x\ y\ a$, under the
call-by-value evaluation strategy adopted by OCaml, never causes an
assertion failure, that is $\forall x,y,a \in \natset.\ \verb|main|\
x\ y\ a \centernot\longrightarrow^* \verb|assert|\ \FALSE$.
By using existing Horn constraint generation methods for call-by-value
functional programs~\cite{Knowles2007,Unno2009}, the relational
verification problem is reduced to the constraint solving problem of
the following Horn clause constraint set $\hmult$:
\[
\left\{
\begin{array}{l}
P(x, 0, 0), \\
P(x, y, x+r) \Leftarrow P(x, y-1, r) \land (y \neq 0),\\
Q(x, 0, a, a), \\
Q(x, y, a, r) \Leftarrow Q(x, y-1, a+x, r) \land (y \neq 0),\\
\bot \Leftarrow P(x, y, r_1) \land Q(x, y, a, r_2) \land (r_1 + a \neq r_2)
\end{array}
\right\}
\]
Here, the predicate variable $P$ (resp. $Q$) represents an inductive
invariant among the arguments and the return value of the function
$\mult$ (resp. $\multacc$).  The first Horn clause $P(x, 0, 0)$ is
generated from the then-branch of the definition of $\mult$ and
expresses that $\mult$ returns $0$ if $0$ is given as the second
argument.  The second clause in $\hmult$, $P(x, y, x+r) \Leftarrow
P(x, y-1, r) \land (y \neq 0)$ is generated from the else-branch and
represents that $\mult$ returns $x+r$ if the second argument $y$ is
non-zero and $r$ is returned by the recursive call $\verb|mult x (y-1)|$.
The other Horn clauses are similarly generated from the then- and
else- branches of $\multacc$ and the assertion in $\main$.  Because
$\hmult$ has a satisfying substitution (i.e., solution)
$\thetamult=\{P \mapsto \lambda(x,y,r).x \times y=r,Q\mapsto
\lambda(x,y,a,r).x \times y + a = r\}$ for the predicate variables $P$
and $Q$, the correctness of the constraint generation
method~\cite{Unno2009} guarantees that the call-by-value evaluation of
$\main\ x\ y\ a$ never causes an assertion failure.

The previous Horn constraint solving methods, however, cannot solve
this kind of constraints that require a relational analysis of
multiple predicates.  To see why, recall the constraint in $\hmult$,
$\bot \Leftarrow P(x, y, r_1) \land Q(x, y, a, r_2) \land (r_1 +
a \neq r_2)$ that asserts the equivalence of $\mult$ and $\multacc$,
where a relational analysis of the two predicates $P$ and $Q$ is
required.  The previous methods, however, analyze each predicate $P$
and $Q$ separately, and therefore must infer nonlinear invariants
$r_1=x \times y$ and $r_2=x \times y + a$ respectively for the
predicate applications $P(x, y, r_1)$ and $Q(x, y, a, r_2)$ to
conclude $r_1+a=r_2$ by canceling $x \times y$, because $x$ and $y$
are the only shared arguments between $P(x, y, r_1)$ and $Q(x, y, a,
r_2)$.  The previous methods can only find solutions that are
expressible by efficiently decidable theories such as the
quantifier-free linear real (\QFLRA) and integer (\QFLIA)
arithmetic\footnote{See \url{http://smt-lib.org/} for the definition
of the theories.}, which are not powerful enough to express the above
nonlinear invariants and the solution $\thetamult$ of $\hmult$.

By contrast, our induction-based Horn constraint solving method can
directly and automatically show that the predicate applications $P(x,
y, r_1)$ and $Q(x, y, a, r_2)$ imply $r_1+a=r_2$ (i.e., $\hmult$ is
solvable), by simultaneously analyzing $P(x, y, r_1)$ and $Q(x, y, a,
r_2)$.  More precisely, our method interprets $P,Q$ as the predicates
inductively defined by the definite clauses (i.e., the clauses whose
head is a predicate application) of $\hmult$, and uses induction on
the derivation of $P(x, y, r_1)$ to prove the conjecture $\forall
x,y,r_1,a,r_2.\left( P(x, y, r_1) \land Q(x, y, a, r_2) \land (r_1 +
a \neq r_2) \Rightarrow \bot \right)$ represented by the goal clause
(i.e., the clause whose head is \emph{not} a predicate application) of
$\hmult$.  Section~\ref{sec:overview} gives an overview of our method
using this running example.

The use of Horn clause constraints, which can be considered as an
Intermediate Verification Language (IVL) common to Horn constraint
solvers and target languages, enables our method to verify relational
specifications across programs written in various paradigms.
Horn clause constraints can naturally axiomatize various advanced
language features including recursive functions that are partial
(i.e., possibly non-terminating), non-deterministic, higher-order,
exception-raising, and over non-inductively defined data types (recall
that $\hmult$ axiomatizes the partial functions $\mult$ and
$\multacc$, and see Section~\ref{sec:hornex} for more examples).
Furthermore, we can automate the axiomatization process by using
program logics such as Hoare logics for imperative and refinement type
systems~\cite{Xi1999,Unno2008,Rondon2008,Unno2009} for functional
programs.  In fact, researchers have developed and made available
tools such as SeaHorn~\cite{Gurfinkel2015} and
JayHorn~\cite{Kahsai2016}, respectively for translating C and Java
programs into Horn clause constraints.
In spite of the expressiveness, Horn clause constraints have a simpler
logical semantics compared to other popular IVLs like
Boogie~\cite{Barnett2006} and Why3~\cite{Bobot2011}.  This simplicity
enabled us to directly apply inductive theorem proving and made the
correctness proof and implementation easier.

In contrast to our induction method based on the logic of predicates
defined by Horn clause constraints, most state-of-the-art automated
inductive theorem provers such as ACL2s~\cite{Chamarthi2011},
Leon~\cite{Suter2011a}, Dafny~\cite{Leino2012},
Zeno~\cite{Sonnex2012}, HipSpec~\cite{Claessen2013}, and
CVC4~\cite{Reynolds2015} are based on logics of pure total functions
over inductively-defined data structures.  Consequently, the
axiomatization of advanced language features and specifications
becomes a non-straightforward process, which often requires users'
manual intervention and possibly has a negative effect on the
automation of induction later.  Thus, our approach complements
automated inductive theorem proving with the expressive power of Horn
clause constraints and, from the opposite point of view, opens the way
to leveraging the achievements of the automated induction community
into Horn constraint solving.

The rest of the paper is organized as follows. In
Section~\ref{sec:overview}, we will give an overview of our
induction-based Horn constraint solving method.
Section~\ref{sec:horn} formalizes Horn constraint solving problems and
shows examples of the reduction from various program verification
problems to Horn constraint solving problems.
Section~\ref{sec:induction} formalizes our constraint solving method
and proves its correctness.  Section~\ref{sec:exp} reports on our
prototype implementation based on the proposed method and the results
of preliminary experiments.  We compare our method with related work
in Section~\ref{sec:related} and conclude the paper with some remarks
on future work in Section~\ref{sec:conc}.

%% file: overview.tex
In this section, we use the Horn constraint set $\hmult$ in
Section~\ref{sec:intro} as a running example to give an overview of
our induction-based Horn constraint solving method.  Our method
interprets the definite clauses (i.e., the clauses whose head is a
predicate application) of a given Horn constraint set as derivation
rules for predicate applications $P(\seq{t})$, which we
call \emph{atoms} henceforth.  For example, the definite clauses
$\dcmult \subseteq \hmult$ are interpreted as the following derivation
rules:
\begin{center}
  \begin{minipage}{0.4\hsize}
    \infrule
      {\models y = 0 \land r = 0}
      {P(x, y, r)}
  \end{minipage}
  \begin{minipage}{0.58\hsize}
    \infrule
      {P(x, y - 1, r - x) \andalso
        \models y \neq 0}
      {P(x, y, r)}
  \end{minipage}
\rulesp

  \begin{minipage}{0.4\hsize}
    \infrule
      {\models y = 0 \land a = r}
      {Q(x, y, a, r)}
  \end{minipage}
  \begin{minipage}{0.58\hsize}
    \infrule
      {Q(x, y - 1, a + x, r) \andalso
        \models y \neq 0}
      {Q(x, y, a, r)}
  \end{minipage}
\end{center}
Here, the heads of the clauses are changed into the uniform
representations $P(x,y,r)$ and $Q(x,y,a,r)$ of atoms over variables.
The above rules inductively define the least predicate interpretation
$\{P \mapsto \set{(x,y,r) \in \intset^3 \mid x \times y = r \land
y \geq 0},Q \mapsto \{(x,y,a,r) \in \intset^4 \mid x \times
y+a=r \land y \geq 0\}\}$ that satisfies the definite clauses
$\dcmult$.
It then follows that a given Horn constraint set has a solution if and
only if all the goal clauses (i.e., the clauses whose head
is \emph{not} an atom) are valid under the interpretation (see
Corollary~\ref{cor:leastmodel} for the proof).  Therefore, constraint
solving of $\hmult$ boils down to the validity checking of the goal
clause
\[
\forall x,y,r_1,a,r_2.\left( P(x, y, r_1) \land Q(x, y, a, r_2) \land (r_1+a \neq r_2) \Rightarrow \bot \right)
\]
under the least predicate interpretation for $P$ and $Q$.

\input{rules_simplified}

To check the validity of such a conjecture, our method uses induction
on the derivation of atoms.
\begin{principle}[Induction on Derivations]
\label{principle:ind}
Let $\mathcal{P}$ be a property on derivations $D$ of atoms.  We then
have $\forall D. \mathcal{P}(D)$ if and only if $\forall
D.\left( \left(\forall D' \prec
D. \mathcal{P}(D')\right) \Rightarrow \mathcal{P}(D) \right)$, where
$D' \prec D$ represents that $D'$ is a strict sub-derivation of $D$.
\end{principle}
Formally, we propose an inductive proof system for deriving judgments
of the form $\D;\Gamma; A; \phi \vdash \bot$, where $\bot$ represents
the contradiction, $\phi$ represents a formula without atoms, $A$
represents a set of atoms, $\Gamma$ represents a set of induction
hypotheses and user-specified lemmas, and
$\D$ represents a set of definite clauses that define the least
predicate interpretation of the predicate variables in $\Gamma$ or
$A$.  Here, $\Gamma$, $A$, and $\phi$ are allowed to have common free
term variables.  The free term variables of a clause in $\D$ have the
scope within the clause, and are considered to be universally
quantified (see Section~\ref{sec:horn} for a formal account).
Intuitively, a judgment $\D;\Gamma; A; \phi \vdash \bot$ means that
under the least predicate interpretation induced by $\D$, the formula
$\bigwedge \Gamma \land \bigwedge A \land \phi \Rightarrow \bot$ is
valid.  For example, consider the following judgment $\Jmult$:
\[
\Jmult \triangleq \dcmult; \emptyset; \set{P(x, y, r_1), Q(x, y, a, r_2)}; (r_1 + a \neq r_2) \vdash \bot
\]
If $\Jmult$ is derivable, $P(x, y, r_1) \land Q(x, y, a, r_2) \land
(r_1+ a \neq r_2) \Rightarrow \bot$ is valid under the least predicate
interpretation induced by $\dcmult$, and hence $\hmult$ has a
solution.

The inference rules for the judgment $\D;\Gamma; A; \phi \vdash \bot$
are shown in Figure~\ref{fig:rule}.  The rules there, however, are too
general and formal for the purpose of providing an overview of the
idea.  Therefore, we defer a detailed explanation of the rules to
Section~\ref{sec:induction}, and here explain a simplified version
shown in Figure~\ref{fig:rule_simplified}, obtained from the complete
version by eliding some conditions and subtleties while retaining the
essence.
The rules are designed to exploit $\Gamma$ and $\D$ for iteratively
updating the current \emph{knowledge} represented by the formula
$\bigwedge A \land \phi$ until a contradiction is implied.
The first rule \rn{Induct} selects an atom $P(\seq{t}) \in A$ and
performs induction on the derivation of the atom by adding a new
induction hypothesis
$\forall \seq{x}.\left( \left(P(\sigma\seq{t}) \prec
P(\seq{t})\right) \land \bigwedge \sigma
A \Rightarrow \neg(\sigma\phi) \right)$ to $\Gamma$.  Here, a map
$\sigma$ is used to generalize the free term variables $\seq{y}$ that
occur in $A$ or $\phi$ (denoted by $\fvs{A} \cup \fvs{\phi}$) into
fresh variables $\seq{x}$, and $P(\sigma\seq{t}) \prec P(\seq{t})$
requires that the derivation of $P(\sigma\seq{t})$ is a strict
sub-derivation of that of $P(\seq{t})$.  The second rule \rn{Unfold}
selects an atom $P(\seq{t}) \in A$, performs a case analysis on the
last rule used to derive the atom, which is represented by a definite
clause in $\D$ of the form $P(\seq{x}) \Leftarrow A' \land \phi'$, and
updates the current knowledge $\bigwedge A \land \phi$ with $\bigwedge
(A \cup \sigma A') \land \phi \land \sigma\phi'$ for
$\sigma=\set{\seq{x} \mapsto \seq{t}}$.
The third rule \rn{Apply$\bot$} selects an induction hypothesis in
$\Gamma$, $\forall \seq{x}.\left( \left(P(\seq{t'}) \prec
P(\seq{t})\right) \land \bigwedge A' \Rightarrow \phi'\right)$, and
tries to find an instantiation $\sigma$ of the quantified variables
$\seq{x}$ such that
\begin{itemize}
\item the instantiated premise $\bigwedge \sigma A'$ of the hypothesis is implied by the current knowledge $\bigwedge A \land \phi$ and
\item the derivation of the atom $P(\sigma\seq{t'}) \in \sigma A'$ to which the hypothesis is being applied is a strict sub-derivation of that of the atom $P(\seq{t})$ on which the induction (that has introduced the hypothesis) has been performed.
\end{itemize}
If such a $\sigma$ is found, the rule updates the current knowledge
with $\bigwedge A \land \phi \land \sigma \phi'$.
The fourth rule \rn{Valid$\bot$} checks whether
$\phi \Rightarrow \bot$ is valid, and if it is the case, closes the
proof branch under consideration.

\begin{figure}[t]
\begin{center}
\infer[(\rn{Induct})]
  {\Jmult}
  {\infer[(\rn{Unfold})]
     {J_0}
     {\infer[(\rn{Unfold})]
        {J_1}
        {\infer[(\rn{Valid$\bot$})]{J_3}{} &
         \infer[(\rn{Valid$\bot$})]{J_4}{}} &
      \infer[(\rn{Unfold})]
        {J_2}
        {\infer[(\rn{Valid$\bot$})]{J_5}{} &
         \infer[(\rn{Apply$\bot$})]
           {J_6}
           {\infer[(\rn{Valid$\bot$})]{J_7}{}}}}}
\begin{align*}
J_0&\triangleq\dcmult; \Gammamult; \Amult_{\emptyset}; r_1 + a \neq r_2 \vdash \bot \\
J_1&\triangleq\dcmult; \Gammamult; \Amult_{\emptyset}; r_1 + a \neq r_2 \land y = 0 \land r_1 = 0 \vdash \bot \\
J_2&\triangleq\dcmult; \Gammamult; \Amult_P; r_1 + a \neq r_2 \land y \neq 0 \vdash \bot \\
J_3&\triangleq\dcmult; \Gammamult; \Amult_{\emptyset}; r_1 + a \neq r_2 \land y = 0 \land r_1 = 0 \land a = r_2 \vdash \bot \\
J_4&\triangleq\dcmult; \Gammamult; \Amult_Q; r_1 + a \neq r_2 \land y = 0 \land r_1 = 0 \land y \neq 0 \vdash \bot \\
J_5&\triangleq\dcmult; \Gammamult; \Amult_P; r_1 + a \neq r_2 \land y \neq 0 \land y = 0 \land a = r_2 \vdash \bot \\
J_6&\triangleq\dcmult; \Gammamult; \Amult_{PQ}; r_1 + a \neq r_2 \land y \neq 0 \vdash \bot \\
J_7&\triangleq\dcmult; \Gammamult; \Amult_{PQ}; r_1 + a \neq r_2 \land y \neq 0 \land r_1 + a = r_2 \vdash \bot \\
\Gammamult &\triangleq \{\forall x', y', r_1', a', r_2'.( (P(x', y', r_1') \prec P(x, y, r_1)) \land \\
& \qquad P(x', y', r_1') \land Q(x', y', a', r_2') \Rightarrow r_1' + a' = r_2') \} \\
\Amult_{\emptyset} &\triangleq \set{P(x, y, r_1), Q(x, y, a, r_2)} \\
\Amult_P &\triangleq \Amult_{\emptyset} \cup \{P(x, y - 1, r_1 - x)\} \\
\Amult_Q &\triangleq \Amult_{\emptyset} \cup \set{Q(x, y - 1, a + x, r_2)} \\
\Amult_{PQ} &\triangleq \Amult_P \cup \set{Q(x, y - 1, a + x, r_2)}
\end{align*}
\end{center}
\caption{The structure of an example derivation of $\Jmult$.}
\label{fig:deriv}
\end{figure}

Figure~\ref{fig:deriv} shows the structure (with side-conditions
omitted) of a derivation of the judgment $\Jmult$, constructed by
using the simplified version of the inference rules.  We below explain
how the derivation is constructed.  First, by performing induction on
the atom $P(x, y, r_1)$ in $\Jmult$ using the rule \rn{Induct}, we
obtain the subgoal $J_0$ with an induction hypothesis $\Gammamult$
added.  We then apply \rn{Unfold} to perform a case analysis on the
last rule used to derive the atom $P(x, y, r_1)$, and obtain the two
subgoals $J_1$ and $J_2$ as the result, because $\Dmult$ has two
clauses with the head that matches with the atom $P(x, y, r_1)$.  The
two subgoals are then discharged as follows.
\begin{itemize}
\item {\bf Subgoal 1:} By performing a case analysis on $Q(x, y, a, r_2)$ in $J_1$ using the rule \rn{Unfold}, we further get two subgoals $J_3$ and $J_4$.  Both $J_3$ and $J_4$ are proved by
the rule \rn{Valid$\bot$} because $\models \phi_3 \Rightarrow \bot$
and $\models \phi_4 \Rightarrow \bot$ hold.
\item {\bf Subgoal 2:} By performing a case analysis on $Q(x, y, a, r_2)$ in $J_2$ using the rule \rn{Unfold}, we obtain two subgoals $J_5$ and $J_6$.  $J_5$ is proved by the rule \rn{Valid$\bot$} because $\models \phi_5 \Rightarrow \bot$ holds.  We then apply the induction hypothesis in $\Gammamult$,
\begin{align*}
&\forall x', y', r_1', a', r_2'.( (P(x', y', r_1') \prec P(x, y, r_1)) \land \\
&\qquad P(x', y', r_1') \land Q(x', y', a', r_2') \Rightarrow r_1' + a' = r_2')
\end{align*}
to the atom $P(x, y - 1, r_1 - x) \in \Amult_{PQ}$ in $J_6$ using the
rule \rn{Apply$\bot$}.  Note that this can be done by using the
quantifier instantiation $\sigma$ defined by
\[
\set{x' \mapsto x, y' \mapsto y-1, r_1' \mapsto r_1-x, a' \mapsto a+x, r_2' \mapsto r_2},
\]
because $\sigma(P(x', y', r_1')) = P(x, y - 1, r_1 - x) \prec P(x, y,
r_1)$ holds and the premise $\sigma(P(x', y', r_1') \land Q(x', y',
a', r_2'))=P(x, y-1, r_1-x) \land Q(x, y-1, a+x, r_2)$ of the
instantiated hypothesis is implied by the current knowledge
$\bigwedge \Amult_{PQ} \land r_1 + a \neq r_2 \land y \neq 0$.  We
thus obtain the subgoal $J_7$, where the $\phi$-part of the knowledge
is updated to
\begin{align*}
& r_1 + a \neq r_2 \land y \neq 0 \land \sigma(r_1' + a' = r_2') \\
&\equiv r_1 + a \neq r_2 \land y \neq 0 \land (r_1 - x) + (a + x) = r_2 \\
&\equiv r_1 + a \neq r_2 \land y \neq 0 \land r_1 + a = r_2.
\end{align*}
Because this implies a contradiction, $J_7$ is finally proved by using
the rule \rn{Valid$\bot$}.
\end{itemize}

To automate proof search in the system, this paper proposes an
SMT-based technique: we use an off-the-shelf SMT solver for checking
whether the current knowledge implies a contradiction (in the
rule \rn{Valid$\bot$}) and whether there is an element of $\Gamma$
that can be used to update the current knowledge, by finding a
quantifier instantiation $\sigma$ (in the rule \rn{Apply$\bot$}).  The
use of an SMT solver provides our method with efficient and powerful
reasoning about data structures supported by SMT, including integers,
real numbers, arrays, algebraic data types (ADTs), and uninterpreted
functions.  There, however, still remain two challenges to be
addressed towards full automation:
\begin{enumerate}

\item {\bf Challenge:}
How to check (in the rule \rn{Apply$\bot$}) the strict sub-derivation
relation $P(\seq{t}') \prec P(\seq{t})$ between the derivation of an
atom $P(\seq{t}')$ to which an induction hypothesis in $\Gamma$ is
being applied, and the derivation of the atom $P(\seq{t})$ on which
the induction has been performed?  Recall that in the above derivation
of $\Jmult$, we needed to check $P(x, y - 1, r_1 - x) \prec P(x, y,
r_1)$ before applying the rule \rn{Apply$\bot$} to $J_6$.

{\bf Our solution:} The formalized rules presented in
Section~\ref{sec:induction} keep sufficient information for checking
the strict sub-derivation relation: we associate each induction
hypothesis in $\Gamma$ with an \emph{induction identifier} $\alpha$,
and each atom in $A$ with a set $M$ of identifiers indicating which
hypotheses can be applied to the atom.  Further details are explained
in Section~\ref{sec:induction}.
\item {\bf Challenge:} In which order should the rules be applied?

{\bf Our solution:} This paper adopts the following simple strategy,
and evaluates it by experiments.
\begin{itemize}
\item Repeatedly apply the rule \rn{Apply$\bot$} if possible, until no new knowledge is obtained.  (Even if the rule does not apply, applications of \rn{Induct} and \rn{Unfold} explained in the following items may make \rn{Apply$\bot$} applicable.)
\item If the current knowledge cannot be updated by using the rule \rn{Apply$\bot$}, select some atom from $A$ in a breadth-first manner, and apply the rule \rn{Induct} to the atom.
\item Apply the rule \rn{Unfold} whenever \rn{Induct} is applied.
\item Try to apply the rule \rn{Valid$\bot$} whenever the $\phi$-part of the knowledge is updated.
\end{itemize}
\end{enumerate}

%% file: rules_simplified.tex
\begin{figure}[t]
\infrule[Induct]
  {P(\seq{t}) \in A \andalso
   \set{\seq{y}}=\fvs{A} \cup \fvs{\phi} \\
   \seq{x} : \fresh \andalso
   \sigma = \set{\seq{y} \mapsto \seq{x}} \\
   \D; \Gamma \cup \set{\forall \seq{x}.\left( \left(P(\sigma\seq{t}) \prec P(\seq{t})\right) \land \bigwedge \sigma A \Rightarrow \neg(\sigma\phi)\right)}; A; \phi \vdash \bot}
  {\D; \Gamma; A; \phi \vdash \bot}
\rulesp

\infrule[Unfold]
  {P(\seq{t}) \in A \andalso
   \sigma=\set{\seq{x} \mapsto \seq{t}} \andalso
   \D; \Gamma; A \cup \sigma A'; \phi \land \sigma \phi' \vdash \bot \\
   (\foreach\ (P(\seq{x}) \Leftarrow A' \land \phi') \in \D)}
  {\D; \Gamma; A; \phi \vdash h}
\rulesp

\infrule[Apply$\bot$]
  {\forall \seq{x}.\left( \left(P(\seq{t'}) \prec P(\seq{t})\right) \land \bigwedge A' \Rightarrow \phi'\right) \in \Gamma \andalso
   \dom{\sigma} = \set{\seq{x}} \\
   P(\sigma \seq{t'}) \prec P(\seq{t}) \andalso
   \models \bigwedge A \land \phi \Rightarrow \bigwedge \sigma A' \\
   \D; \Gamma; A; \phi \land \sigma \phi' \vdash \bot}
  {\D; \Gamma; A; \phi \vdash \bot}
\rulesp

\infrule[Valid$\bot$]
  {\models \phi \Rightarrow \bot}
  {\D; \Gamma; A; \phi \vdash \bot}
\rulesp

\caption{A simplified version of the inference rules in Figure~\ref{fig:rule} for the judgment $\D; \Gamma; A; \phi \vdash \bot$.}
\label{fig:rule_simplified}
\end{figure}

%% file: horn.tex
This section formalizes Horn constraint solving problems and proves
the correctness of our reduction from Horn constraint solving to
inductive theorem proving in Corollary~\ref{cor:leastmodel}.
Section~\ref{sec:hornex} also shows example Horn constraint solving
problems reduced from (relational) verification problems of programs
that use various advanced language features, including higher-order
functions and exceptions.

The syntax of Horn Clause Constraint Sets (HCCSs) over the theory
$\LINT$ of quantifier-free linear integer arithmetic is defined by
\begin{align*}
(\mbox{HCCSs})\ \h &\coleq \set{\hc_1, \ldots, \hc_m}\\
(\mbox{Horn clauses})\ \hc &\coleq h \Leftarrow b \\
(\mbox{heads})\ h &\coleq P(\seq{t}) \mid \bot \\
(\mbox{bodies})\ b &\coleq P_1(\seq{t}_1) \land \cdots \land P_m(\seq{t}_m) \land \phi \\
(\LINT\mbox{-formulas})\ \phi &\coleq t_1 \leq t_2 \mid \top \mid \bot \mid \neg \phi \mid \phi_1 \land \phi_2 \mid \phi_1 \lor \phi_2 \\
(\LINT\mbox{-terms})\ t &\coleq x \mid n \mid t_1 + t_2
\end{align*}
Here, the meta-variables $P$ and $x$ respectively represent predicate
variables and term variables, and $\seq{t}$ represents a sequence of
terms $t_1,\dots,t_m$.  We write the arity of $P$ as $\arity{P}$.
Note that, in the syntax of $\LINT$-formulas, linear inequalities $t_1
\leq t_2$ can be used to encode $t_1 < t_2$, $t_1 = t_2$, and $t_1
\neq t_2$.  For example, $t_1 < t_2$ is encoded as $t_1 + 1 \leq
t_2$.  The formula $\top$ (resp. $\bot$) represents the tautology
(resp. the contradiction).  We here restrict ourselves to $\LINT$ for
simplicity, although our induction-based Horn constraint solving
method formalized in Section~\ref{sec:induction} supports constraints
over whatever background theories supported by the underlying SMT
solver, including the theories of nonlinear arithmetics, algebraic
data structures, and uninterpreted function symbols as shown in
Section~\ref{sec:hornex}.

\subsection{Notation for HCCSs}
\label{sec:notation}

A \emph{Horn clause constraint set} $\h$ is a finite set
$\set{\hc_1, \ldots, \hc_m}$ of Horn clauses.  A \emph{Horn clause}
$h \Leftarrow b$ consists of a head $h$ and a body $b$.  We often
abbreviate a Horn clause $h \Leftarrow \top$ as $h$.
We write $\pvs{\hc}$ for the set of the predicate variables that occur
in $\hc$ and define $\pvs{\h}=\bigcup_{\hc \in \h}\pvs{\hc}$.
Similarly, we write $\fvs{\hc}$ for the set of the term variables in
$\hc$ and define $\fvs{\h}=\bigcup_{\hc \in \h}\fvs{\hc}$.  We assume
that for any $\hc_1,\hc_2 \in \h$, $\hc_1 \neq \hc_2$ implies
$\fvs{\hc_1}
\cap \fvs{\hc_2}=\emptyset$.  We write $\h \restriction_P$ for the set
of Horn clauses in $\h$ of the form $P(\seq{t}) \Leftarrow b$.  We
define $\h(P)=\lambda \seq{x}.\exists \seq{y}.\bigvee_{i=1}^m
(b_i \land \seq{x}=\seq{t}_i)$ if
$\h \restriction_P=\set{P(\seq{t}_i) \Leftarrow
b_i}_{i \in \set{1,\dots,m}}$ where
$\set{\seq{y}}=\fvs{\h \restriction_P}$ and
$\set{\seq{x}} \cap \set{\seq{y}}=\emptyset$.  By using $\h(P)$, an
HCCS $\h$ is logically interpreted as the formula
\[
\bigwedge_{P \in \pvs{\h}} \forall \seq{x}_P.\left( \h(P)(\seq{x}_P) \Rightarrow P(\seq{x}_P) \right),
\]
where $\seq{x}_P=x_1,\dots,x_{\arity{P}}$.
A Horn clause with the head of the form $P(\seq{t})$ (resp. $\bot$) is
called a \emph{definite} clause (resp. a \emph{goal} clause).  We
write $\defof{\h}$ (resp. $\goalof{\h}$) for the subset of $\h$
consisting of only the definite (resp. goal) clauses.  Note that $\h
= \defof{\h}
\cup \goalof{\h}$ and $\defof{\h} \cap \goalof{\h}=\emptyset$.

\subsection{Predicate Interpretation}

A \emph{predicate interpretation} $\rho$ for an HCCS $\h$ is a map
from each predicate variable $P \in \pvs{\h}$ to a subset of
$\intset^{\arity{P}}$.  We write the domain of $\rho$ as $\dom{\rho}$.
We write $\rho_1 \subseteq\rho_2$ if $\rho_1(P) \subseteq \rho_2(P)$
for all $P \in \pvs{\h}$.  We call an interpretation $\rho$
a \emph{solution of $\h$} and write $\rho \models \h$ if
$\rho \models \hc$ holds for all $\hc \in \h$.  For example,
$\rhomult= \{P \mapsto \set{ (x, y, r) \in \intset^3 \mid x \times y =
r }, Q \mapsto \set{ (x, y, a, r) \in \intset^4 \mid x \times y + a =
r}\}$ is a solution of the HCCS $\hmult$ in Section~\ref{sec:intro}.

\begin{definition}[Horn Constraint Solving Problems]
A \emph{Horn constraint solving problem} is the problem of checking
whether a given HCCS $\h$ has a solution.
\end{definition}

We now establish the reduction from Horn constraint solving to
inductive theorem proving, which is the foundation of our
induction-based Horn constraint solving method.

The definite clauses $\defof{\h}$ are considered to inductively define
the \emph{least predicate interpretation} for $\h$ as the least
fixed-point $\leastmodel$ of the following function on predicate
interpretations.
\[
F_{\h}(\rho) = \left\{
  P \mapsto \left\{(\seq{x}) \in \intset^{\arity{P}} \relmiddle{|}
                   \rho \models \h(P)(\seq{x})\right\} \relmiddle{|}
  P \in \dom{\rho}
  \right\}
\]
Because $F_{\h}$ is continuous~\cite{Jaffar1994}, the least
fixed-point $\leastmodel$ of $F_{\h}$ exists.  Furthermore, we can
express it as
$$
\leastmodel=\bigcup_{i \in \natset} F_{\h}^i(\set{P \mapsto \emptyset \mid P \in \pvs{\h}}),
$$
where $F_{\h}^i$ means $i$-times application of $F_{\h}$.
It immediately follows that the least predicate interpretation
$\leastmodel$ is a solution of $\defof{\h}$ because any fixed-point of
$F_{\h}$ is a solution of $\defof{\h}$.  Furthermore, $\leastmodel$ is
the least solution.  Formally, we can prove the following proposition.
\begin{proposition}
\label{prp:leastmodel}
$\leastmodel \models \defof{\h}$ holds, and for all $\rho$ such that
$\rho \models \defof{\h}$, $\leastmodel \subseteq \rho$ holds.
\end{proposition}

On the other hand, the goal clauses $\goalof{\h}$ are considered as
specifications of the least predicate interpretation $\leastmodel$.
As a corollary of Proposition~\ref{prp:leastmodel}, it follows that $\h$ has
a solution if and only if $\leastmodel$ satisfies the specifications
$\goalof{\h}$.
\begin{corollary}
\label{cor:leastmodel}
$\rho \models \h$ for some $\rho$ if and only if
$\leastmodel \models \goalof{\h}$
\end{corollary}
In Section~\ref{sec:induction}, we present an induction-based method
for proving $\leastmodel \models \goalof{\h}$.

\input{horn_ex}

%% file: horn_ex.tex
\subsection{Examples Reduced from Program Verification Problems}
\label{sec:hornex}

This section shows example Horn constraint solving problems reduced
from (relational) verification problems of programs that use advanced
language features such as algebraic data structures, higher-order
functions, and exceptions.
The reduction used in this section is mostly based on an existing Horn
constraint generation method~\cite{Unno2009} for an ML-like (i.e.,
call-by-value, statically-typed, and higher-order) functional
language.  The method can be used to reduce a given assertion safety
verification problem defined below into a Horn constraint solving
problem.
\begin{definition}[Assertion Safety Verification Problems]
\label{def:prob}
An assertion safety verification problem of a given functional
program, with a special function $\main$ of the ordinary ML type
$\TINT \to \dots \to \TINT \to \TUNIT$, is the problem of deciding
whether
\[\forall n_1,\dots,n_m \in \intset.\ \main\ n_1\dots n_m\ \notreds\ \ASSERT\ \FALSE,\]
where $\longrightarrow$ is the one-step evaluation relation.  We call
the program \emph{safe} if this property holds, and \emph{unsafe}
otherwise.
\end{definition}
The constraint generation method is based on refinement
types~\cite{Xi1999}, which are used internally to express value
dependent inductive invariants and specifications of the program.  The
following theorem states the soundness of the reduction.
\begin{theorem}[Soundness~\cite{Unno2009}]
Let $\h$ be the HCCS generated from a program $D$.  If there exists a
solution $\rho$ of $\h$, then $D$ is safe.
\label{verify_theo}
\end{theorem}

We now show example Horn constraints generated by the method.  The
partial recursive functions shown in Section~\ref{sec:intro} are
automatically axiomatized using refinement types as follows.
\begin{example}
\label{ex:mult}
Recall the program $\Dmult$ in Section~\ref{sec:intro}.  The
constraint generation method first prepares the following refinement
type templates for the functions in $\Dmult$.
\[
\begin{array}{rl}
\mathtt{mult}:& (x: \TINT) \to (y: \TINT) \to \set{r: \TINT \mid P(x, y, r)} \\
\mathtt{mult\_acc}:& (x: \TINT) \to (y: \TINT) \to (a: \TINT) \to \\
                   & \set{r: \TINT \mid Q(x, y, a, r)}
\end{array}
\]
Here, the predicate variable $P$ (resp. $Q$) represents an inductive
invariant among the arguments and the return value of the function
$\mult$ (resp. $\multacc$).  The constraint generation method then
type-checks the program against the type templates, and obtains the
Horn constraint set $\hmult$ in Section~\ref{sec:intro}, which has a
solution if and only if the program is typable under a refinement type
system.  The refinement type system guarantees that if a given program
is typable, the evaluation of $\verb|main|\ n$ never causes an
assertion failure for any integer $n$.
\qed
\end{example}

\begin{example}
Consider the following program.
\begin{alltt}
let rec sum n =
  if n < 0 then n + sum (n + 1)
  else if n = 0 then 0 else n + sum (n - 1)
let rec sum_acc n a =
  if n < 0 then sum_acc (n + 1) (a + n)
  else if n = 0 then a else sum_acc (n - 1) (a + n)
let main n a = assert(sum n + a = sum_acc n a)
\end{alltt}
In a similar manner to Example~\ref{ex:mult}, we obtain the following
Horn constraint set:
\[
\left\{
\begin{array}{l}
P(0, 0), \\
P(x, r + x) \Leftarrow P(x + 1, r) \land x < 0, \\
P(x, r + x) \Leftarrow P(x - 1, r) \land x > 0, \\
Q(0, a, a), \\
Q(x, a, r) \Leftarrow Q(x + 1, a + x, r) \land x < 0, \\
Q(x, a, r) \Leftarrow Q(x - 1, a + x, r) \land x > 0, \\
\bot \Leftarrow P(x,r_1) \land Q(x,a,r_2) \land r_1+a \neq r_2
\end{array}
\right\}
\]
Here, the predicate variable $P$ (resp. $Q$) represents an inductive
invariant among the arguments and the return value of the function
$\SUM$ (resp. $\SUMACC$).  Here, suppose that the main function is
replaced by
\begin{alltt}
let main n =
  if n >= 0 then assert (2 * sum n = n * (n + 1))
\end{alltt}
We then obtain the following goal clause over the nonlinear integer
arithmetic instead:
\[
\begin{array}{l}
\bot \Leftarrow P(x,r) \land x \geq 0 \land 2 \times r \neq x \times (x + 1)
\end{array}
\]
\qed
\end{example}

The method can automatically axiomatize complex recursive functions on
integers.
\begin{example}
Consider the following program with complex recursion.
\begin{alltt}
let rec mc91 x =
  if x > 100 then x - 10 else mc91 (mc91 (x + 11))
let main x = if x <= 101 then assert(mc91 x = 91)
\end{alltt}
By using the refinement type template
\[
\begin{array}{rl}
\mathtt{mc91}:& (x: \TINT) \to \set{r: \TINT \mid P(x, r)} \\
\end{array}
\]
the constraint generation method returns the following Horn constraint
set:
\[
\left\{
\begin{array}{l}
P(x, x - 10) \Leftarrow x > 100, \\
P(x, s) \Leftarrow P(x + 11, r) \land P(r, s) \land x \leq 100, \\
\bot \Leftarrow P(x,r) \land x \leq 101 \land r \neq 91
\end{array}
\right\}
\]
Here, the predicate variable $P$ represents an inductive invariant
among the arguments and the return value of the function $\MCC$.
\qed
\end{example}

Our method can also handle recursive functions on non-inductively
defined data types such as real numbers.
\begin{example}
Consider the following program that models a dynamical system from
\cite{Kupferschmid2011}.
\begin{alltt}
let rec dyn_sys vc =
  let fa = 0.5418 *. vc *. vc in (* the force control *)
  let fr = 1000. -. fa in
  let ac = 0.0005 *. fr in
  let vc' = vc +. ac in
  assert (vc' < 49.61); (* the safety velocity *)
  dyn_sys vc'
let main () = dyn_sys 0. (* the initial velocity *)
\end{alltt}
By using the refinement type template
\[
\begin{array}{rl}
\mathtt{dyn\_sys}:& \set{x: \TREAL \mid P(x)} \to \TUNIT \\
\end{array}
\]
the constraint generation method returns the following Horn constraint
set:
\[
\left\{
\begin{array}{l}
  P(\mathit{vc}') \Leftarrow P(\mathit{vc}) \land
  \mathit{fa} = 0.5418 \times \mathit{vc} \times \mathit{vc} \land
  \mathit{fr} = 1000 - \mathit{fa} \land \\
  \qquad\qquad\ \mathit{ac} = 0.0005 \times \mathit{fr} \land
  \mathit{vc}' = \mathit{vc} + \mathit{ac} \land
  \mathit{vc}' < 49.61, \\
  P(0), \\
  \bot \Leftarrow P(\mathit{vc}) \land
  \mathit{fa} = 0.5418 \times \mathit{vc} \times \mathit{vc} \land
  \mathit{fr} = 1000 - \mathit{fa} \land \\
  \qquad\ \mathit{ac} = 0.0005 \times \mathit{fr} \land
  \mathit{vc}' = \mathit{vc} + \mathit{ac} \land
  \mathit{vc}' \geq 49.61
\end{array}
\right\}
\]
Here, the predicate variable $P$ represents an inductive invariant on
the argument of the function $\mathtt{dyn\_sys}$.
\qed
\end{example}

The constraint generation method can handle functional programs that
manipulate user-defined algebraic data structures.
\begin{example}
Consider the following program that manipulates lists.
\begin{alltt}
type list = Nil | Cons of int * list

let rec append l ys = match l with
  | Nil -> ys
  | Cons(x, xs) -> Cons(x, append xs ys)
let rec drop n l = match l with
  | Nil -> Nil
  | Cons(x, xs) ->
    if n = 0 then Cons(x, xs) else drop (n - 1) xs
let rec take n l = match l with
  | Nil -> Nil
  | Cons(x, xs) ->
    if n = 0 then Nil else Cons(x, take (n - 1) xs)
let main n xs =
  assert(append (take n xs) (drop n xs) = xs)
\end{alltt}
By using the refinement type templates
\[
\begin{array}{rl}
\mathtt{append}:& (x: \TLIST) \to (y: \TLIST) \to \set{r: \TLIST \mid P(x, y, r)} \\
\mathtt{drop}:& (x: \TINT) \to (y: \TLIST) \to \set{r: \TLIST \mid Q(x, y, r)} \\
\mathtt{take}:& (x: \TINT) \to (y: \TLIST) \to \set{r: \TLIST \mid R(x, y, r)}
\end{array}
\]
the constraint generation method returns the following Horn constraint
set over the theory of algebraic data structures:
\[
\left\{
\begin{array}{l}
P(\NIL, l_2, l_2), \\
P(\CONS(x,l), l_2, \CONS(x,r)) \Leftarrow P(l,l_2,r), \\
Q(n,\NIL,\NIL), \\
Q(n,\CONS(x,l'),\CONS(x,l')) \Leftarrow n = 0, \\
Q(n,\CONS(x,l'),r) \Leftarrow Q(n-1,l',r) \land n \neq 0, \\
R(n,\NIL,\NIL), \\
R(n,\CONS(x,l'),\NIL) \Leftarrow n = 0, \\
R(n,\CONS(x,l'),\CONS(x,r)) \Leftarrow Q(n-1,l',r) \land n \neq 0, \\
\bot \Leftarrow P(n,l,r_1) \land Q(n,l,r_2) \land R(r_1,r_2,r) \land r \neq l
\end{array}
\right\}
\]
\qed
\end{example}

The method can also axiomatize higher-order functions into Horn clause
constraints automatically.
\begin{example}
Consider the following higher-order program.
\begin{alltt}
type list = Nil | Cons of int * list

let rec sum_list l = match l with
  | Nil -> 0
  | Cons(x, xs) -> x + sum_list xs
let rec fold_left f s l = match l with
  | Nil -> s
  | Cons(x, xs) -> fold_left f (f s x) xs
let plus x y = x + y
let main l = assert(sum_list l = fold_left plus 0 l)
\end{alltt}
By using the refinement type templates
\[ 
\begin{array}{rl}
\mathtt{sum\_list}:& (x: \TLIST) \to \set{r: \TINT \mid P(x, r)} \\
\mathtt{fold\_left}:& (f:(a: \TINT) \to (b: \TINT) \to \set{c:\TINT \mid Q(a,b,c)}) \\
                    &\!\!\to (x:\TINT) \to (y:\TLIST) \to \set{z:\TINT \mid R(x,y,z)} \\
\mathtt{plus}:& (x: \TINT) \to (y: \TLIST) \to \set{r: \TLIST \mid S(x, y, r)}
\end{array}
\]
the constraint generation method returns the following Horn constraint
set over the theories of linear integer arithmetic and algebraic data
structures:
\[
\left\{
\begin{array}{l}
P(\NIL, 0), \\
P(\CONS(x,l'), x+r) \Leftarrow P(l',r), \\
R(s,\NIL,s), \\
R(s,\CONS(x,l'),r') \Leftarrow Q(s,x,r) \land R(r,l',r'), \\
S(x,y,x+y), \\
Q(x,y,z) \Leftarrow S(x,y,z), \\
\bot \Leftarrow P(l,r_1) \land R(0,l,r_2) \land r_1 \neq r_2
\end{array}
\right\}
\]
\qed
\end{example}

The method can also axiomatize recursive functions that may raise
exceptions into Horn clause constraints.
\begin{example}
Consider the following higher-order program that manipulates lists and
possibly raises and catches exceptions.
\begin{alltt}
exception Not_found
type int_option = None | Some of int

let rec find p l = match l with
  | [] -> raise Not_found
  | x::xs -> if p x then x else find p xs
let rec find_opt p l = match l with
  | [] -> None
  | x::xs -> if p x then Some x else find_opt p xs
let main p l = try find_opt p l = Some (find p l)
               with Not_found -> find_opt p l = None
\end{alltt}
Here, \verb|find| and \verb|find_opt| respectively use the exception
\verb|Not_found| and the option type \verb|int_option|, defined
respectively in the first and the second lines, for finding the first
element of the list \verb|l| satisfying the predicate
$\mathtt{p}:\TINT \to \TBOOL$.  The constraint generation method
cannot directly handle this program because the underlying refinement
type system does not support exceptions.  Note, however, that we can
mechanically transform the program into the following one by
eliminating exceptions using a selective CPS
transformation~\cite{Sato2013}.
\begin{alltt}
type exc = Not_found
type int_option = None | Some of int

let rec find p l ok ex = match l with
  | [] -> ex Not_found
  | x::xs -> if p x then ok x else find p xs ok ex
let rec find_opt p l = match l with
  | [] -> None
  | x::xs -> if p x then Some x else find_opt p xs
let main p l =
  find p l (fun x -> assert (find_opt p l = Some x))
   (fun Not_found -> assert (find_opt p l = None))
\end{alltt}
Note here that two function arguments $\mathtt{ok}$ and $\mathtt{ex}$
of the ordinary ML type $\TINT \to \TUNIT$, which respectively
represent continuations for normal and exceptional cases, are added to
the function $\mathtt{find}$.
The constraint generation method then prepares the following
refinement type templates:
\[
\begin{array}{rl}
\mathtt{find}:& (x: \TINT \to \TBOOL) \to (y: \TLIST) \to \\
  &(\set{z: \TINT \mid P_{ok}(x, y, z)} \to \TUNIT) \to \\
  &(\set{w: \mathtt{exc} \mid P_{ex}(x, y, w)} \to \TUNIT) \to \TUNIT \\ 
\mathtt{find\_opt}:& (x: \TINT \to \TBOOL) \to (y: \TLIST) \to \\
  &\set{z: \mathtt{int\_option} \mid Q(x, y, z)}
\end{array}
\]
Here, the predicate variable $P_{ok}$ represents invariants among the
first and the second arguments of \verb|find| and the argument of the
third argument \verb|ok| of \verb|find|.  The predicate variable
$P_{ex}$ represents invariants among the first and the second
arguments of \verb|find| and the argument of the fourth argument
\verb|ex| of \verb|find|.  The predicate variable $Q$ represents
invariants among the arguments and the return value of
\verb|find_opt|.  The constraint generation method then obtains the
following Horn constraint set over the theories of linear integer
arithmetic, algebraic data structures, and uninterpreted function
symbols:
\[
\left\{
\begin{array}{l}
 P_{ok}(p, x::xs, r) \Leftarrow p\ x = \top, \\
 P_{ok}(p, x::xs, r) \Leftarrow P_{ok}(p, xs, r) \land p\ x = \bot, \\
 P_{ex}(p, [\ ], \NOTFOUND), \\
 P_{ex}(p, x::xs, r) \Leftarrow P_{ex}(p, xs, r) \land p\ x = \bot, \\
 Q(p, [\ ], \NONE), \\
 Q(p, x::xs, \SOME\ x) \Leftarrow p\ x = \top, \\
 Q(p, x::xs, r) \Leftarrow Q(p, xs, r) \land p\ x = \bot, \\
 \bot \Leftarrow P_{ok}(p, l, r_1) \land Q(p, l, r_2) \land \SOME\ r_1 \neq r_2 \\
 \bot \Leftarrow P_{ex}(p, l, r_1) \land r_1 \neq \NOTFOUND \\
 \bot \Leftarrow P_{ex}(p, l, \NOTFOUND) \land Q(p, l, r_2) \land \NONE \neq r_2 \\
\end{array}
\right\}
\]
Here, $p$ is an uninterpreted function symbol, which is essential for
the success of Horn constraint solving here because we need to express
the fact that the multiple occurrences of $p\ x$ in the body of
different clauses return the same value if the same function is passed
as $p$.  \qed
\end{example}

Our method also supports demonic non-determinism.
\begin{example}
Consider the following higher-order program that calls \verb|rand_int|
to generate random integers.
\begin{alltt}
let rec randpos dummy =
  let n = rand_int () in
  if n > 0 then n else randpos dummy
let rec sum_fun f n =
  if n = 0 then f 0
  else f n + sum_fun f (n - 1)
let main n = assert (sum_fun randpos n > 0)
\end{alltt}
Note that the specification is satisfied because the function
\verb|randpos| never returns a non-positive integer.  By using the
refinement type templates
\[
\begin{array}{rl}
\mathtt{randpos}:& (x: \TINT) \to \set{y: \TLIST \mid P(x, y)} \\
\mathtt{sum\_fun}:& (f:(a: \TINT) \to \set{b:\TINT \mid Q(a,b)}) \to \\
                  & (x:\TINT) \to \set{y:\TINT \mid R(f,x,y)}
\end{array}
\]
we obtain the following Horn constraint set:
\[
\left\{
\begin{array}{l}
P(x,y) \Leftarrow y > 0, \\
P(x,y) \Leftarrow P(x,y) \land y \leq 0, \\
Q(a,b) \Leftarrow Q(x,r) \land Q(a,b) \land x \neq 0, \\
Q(a,b) \Leftarrow P(a,b), \\
R(f,0,y) \Leftarrow Q(0,y), \\
R(f,x,r_1 + r_2) \Leftarrow \\
\qquad Q(0,y) \land Q(x,r_1) \land R(f,x-1,r_2) \land x \neq 0, \\
\bot \Leftarrow Q(\mathit{randpos},x,y) \land y \leq 0
\end{array}
\right\}
\]
\qed
\end{example}

Our method based on Horn clause constraints is not limited to
relational verification of functional programs.  By combining the
constraint generation tools for C~\cite{Gurfinkel2015} and
Java~\cite{Kahsai2016}, we can axiomatize relational verification
problems across functional, imperative, object-oriented, and, of
course, (constraint) logic programs into Horn clause constraints.
\begin{example}
Consider the following C program.
\begin{alltt}
  int mult(int x, int y) \{
      int r=0; while(y != 0) {r = r + x; y = y - 1;}
      return r;
  \}
\end{alltt}
Using the Hoare logic, we obtain the following Horn constraint set:
\[
\left\{
\begin{array}{l}
 I(x, y, r) \Leftarrow r = 0, \\
 I(x, y-1, r + x) \Leftarrow I(x, y, r) \land y \neq 0, \\
 R(x, y, r) \Leftarrow I(x, y, r) \land y = 0
\end{array}
\right\}
\]
Here, the predicate variable $I$ represents the loop invariant of the
while loop, and $R$ represents invariants among the arguments $x, y$
and the return value $r$ of the procedure \verb|mult|.
The goal clause $\bot \Leftarrow P(x,y,r_1) \land R(x,y,r_2) \land r_1
\neq r_2$, with the predicate $P$ defined by $\hmult$ in
Section~\ref{sec:intro}, represents the equivalence of C and OCaml
implementations of \verb|mult|.  \qed
\end{example}

It is also worth mentioning here that there have also been proposed
techniques for reducing verification problems of multi-threaded
programs~\cite{Gupta2011a,Grebenshchikov2012} and functional programs
with the call-by-need evaluation strategy~\cite{Vazou2014} into Horn
constraint solving problems.  Angelic
non-determinism~\cite{Hashimoto2015a} and temporal
specifications~\cite{Beyene2013} can also be automatically axiomatized
into Horn clause constraints extended with existentially quantified
heads.

%% file: induction.tex
As explained in Section~\ref{sec:overview}, our method is based on the
reduction from Horn constraint solving into inductive theorem proving.
The correctness of the reduction is established by
Corollary~\ref{cor:leastmodel} in Section~\ref{sec:horn}.  The
remaining task is to develop an automated method for proving the
inductive conjectures obtained from Horn clause constraints.  To this
end, Section~\ref{sec:proof_system} formalizes our inductive proof
system tailored to Horn constraint solving and proves its correctness.
Section~\ref{sec:auto_induct} discusses how to automate proof search
in the system using an SMT solver.

\subsection{Inductive Proof System}
\label{sec:proof_system}

\input{rules}

We formalize a general and more elaborate version of the inductive
proof system explained in Section~\ref{sec:overview}.  A judgment of
the extended system is of the form $\D; \Gamma; A; \phi \vdash h$,
where $\D$ is a set of definite clauses and $\phi$ represents a
formula without atoms.  We here assume that $\D(P)$ is defined
similarly as $\h(P)$.
The asserted proposition $h$ on the right is now allowed to be an atom
$P(\seq{t})$ instead of $\bot$.  For deriving such judgments, we will
introduce new rules \rn{Fold} and \rn{Valid$P$} later in this section.
$\Gamma$ represents a set $\set{(g_1,A_1,\phi_1,h_1), \dots,
(g_m,A_m,\phi_m,h_m)}$ consisting of user-specified lemmas and
induction hypotheses, where $g_i$ is either $\bullet$ or
$\alpha \triangleright P(\seq{t})$.
$(\bullet,A,\phi,h) \in \Gamma$ represents the user-specified lemma
\[\forall \seq{x}.\left( \bigwedge A \land \phi \Rightarrow h \right)
\ \mbox{ where }\set{\seq{x}}=\fvs{A,\phi,h},\] while
$(\alpha \triangleright P(\seq{t}),A,\phi,h) \in \Gamma$ represents the
induction hypothesis
\begin{align*}
\forall \seq{x}.\left( \left(P(\seq{t}) \prec P(\seq{t}')\right) \land \bigwedge A \land \phi \Rightarrow h\right) \\
\mbox{where }\set{\seq{x}}=\fvs{P(\seq{t}),A,\phi,h}
\end{align*}
that has been introduced by induction on the derivation of the atom
$P(\seq{t}')$.  Here, $\alpha$ represents the \emph{induction
identifier} assigned to the application of induction that has
introduced the hypothesis.
Note that $h$ on the right-hand side of $\Rightarrow$ is now allowed
to be an atom of the form $Q(\seq{t})$.  We will introduce a new
rule \rn{Apply$P$} later in this section for using such lemmas and
hypotheses to obtain new knowledge.
$A$ is also extended to be a set
$\set{\mkk{P_1}{M_1}{\alpha_1}(\seq{t}_1), \dots, \mkk{P_m}{M_m}{\alpha_m}(\seq{t}_m)}$
of annotated atoms.  Each element $\mkk{P}{M}{\alpha}(\seq{t})$ has
two annotations:
\begin{itemize}
\item an induction identifier $\alpha$ indicating that the induction with the
identifier $\alpha$ is performed on the atom by the rule \rn{Induct}.
If the rule \rn{Induct} has never been applied to the atom, $\alpha$
is set to be a special identifier denoted by $\circ$.
\item a set of induction identifiers $M$ indicating that if $\alpha' \in M$,
the derivation $D$ of the atom $\mkk{P}{M}{\alpha}(\seq{t})$ satisfies
$D \prec D'$ for the derivation $D'$ of the atom $P(\seq{t}')$ on
which the induction with the identifier $\alpha'$ is performed.  Thus,
an induction hypothesis $(\alpha' \triangleright
P(\seq{t}'),A',\phi',h') \in \Gamma$ can be applied to the atom
$\mkk{P}{M}{\alpha}(\seq{t}) \in A$ only if $\alpha' \in M$ holds.
\end{itemize}
Note that we use these annotations only for guiding inductive proofs
and $\mkk{P}{M}{\alpha}(\seq{t})$ is logically equivalent to
$P(\seq{t})$.  We often omit these annotations when they are clear
from the context.

The inference rules for the judgment $\D; \Gamma; A; \phi \vdash h$
are defined in Figure~\ref{fig:rule}.
The rule \rn{Induct} selects an atom $\mkk{P}{M}{\circ}(\seq{t}) \in
A$ and performs induction on the derivation of the atom.  This rule
generates a fresh induction identifier $\alpha \neq \circ$, adds a new
induction hypothesis $(\alpha \triangleright P(\seq{t}),A,\phi,h)$ to
$\Gamma$, and replaces the atom $\mkk{P}{M}{\circ}(\seq{t})$ with the
annotated one $\mkk{P}{M}{\alpha}(\seq{t})$ for remembering that the
induction with the identifier $\alpha$ is performed on it.
The rule \rn{Unfold} selects an atom $\mkk{P}{M}{\alpha}(\seq{t}) \in
A$ and performs a case analysis on the last rule
$P(\seq{t}) \Leftarrow \phi_i \land \bigwedge A_i$ used to derive the
atom.  As the result, the goal is broken into $m$-subgoals if there
are $m$ rules possibly used to derive the atom.  The rule adds
$\mkk{A_i}{M \cup \set{\alpha}}{\circ}$ and $\phi_i$ respectively to
$A$ and $\phi$ in the $i$-th subgoal, where $\mkk{A}{M}{\alpha}$
represents $\left\{\mkk{P}{M}{\alpha}(\seq{t}) \relmiddle{|}
P(\seq{t}) \in A\right\}$.  Note here that each atom in $A_i$ is
annotated with $M \cup \set{\alpha}$ because the derivation of the
atom $A_i$ is a strict sub-derivation of that of the atom
$\mkk{P}{M}{\alpha}(\seq{t})$ on which the induction with the
identifier $\alpha$ has been performed.  If $\alpha=\circ$, it is the
case that the rule \rn{Induct} has never been applied to the atom
$\mkk{P}{M}{\alpha}(\seq{t})$ yet.
The rules \rn{Apply$\bot$} and \rn{Apply$P$} select
$(g,A',\phi',h) \in \Gamma$, which represents a user-specified lemma
if $g=\bullet$ and an induction hypothesis otherwise, and try to add
new knowledge respectively to the $\phi$- and the $A$-part of the
current knowledge: the rules try to find an instantiation $\sigma$ for
the free term variables in $(g,A',\phi',h)$, which are considered to
be universally quantified, and then use $\sigma(g,A',\phi',h)$ to
obtain new knowledge.
Contrary to the rule \rn{Unfold}, the rule \rn{Fold} tries to use a
definite clause $P(\seq{t}) \Leftarrow \phi' \land \bigwedge
A' \in \D$ from the body to the head direction: \rn{Fold} tries to
find $\sigma$ such that $\sigma(\phi' \land \bigwedge A')$ is implied
by the current knowledge, and update it with $P(\sigma\seq{t})$.
This rule is useful when we check the correctness of user specified
lemmas.
The rule \rn{Valid$\bot$} checks if $\phi$ is unsatisfiable, while the
rule \rn{Valid$P$} checks if the asserted proposition $P(\seq{t})$ on
the right-hand side of the judgment is implied by the current
knowledge $\bigwedge A \land \phi$.

Given a Horn constraint solving problem $\h$, our method reduces the
problem into an inductive theorem proving problem as follows.
For each goal clause in $\goalof{\h} = \set{\bigwedge
A_i \land \phi_i \Rightarrow \bot}_{i=1}^m$, we check the judgment
$\defof{\h};\emptyset; \mkk{A_i}{\emptyset}{\circ}; \phi_i \vdash \bot$
is derivable by the inductive proof system.  Here, each atom in $A_i$
is initially annotated with $\emptyset$ and $\circ$.

We now prove the correctness of our method, which follows from the
soundness of the inductive proof system.  To state the soundness, we
first define $\sem{\Gamma}{A}{k}$, which represents the conjunction of
user-specified lemmas and induction hypotheses in $\Gamma$
instantiated for the atoms occurring in the $k$-times unfolding of
$A$.
\begin{align*}
\sem{\Gamma}{A}{0} &= \sembrack{\Gamma}(A), \\
\sem{\Gamma}{A}{k+1} &=
\sembrack{\Gamma}(A) \land
\bigwedge_{\mkk{P}{M}{\alpha}(\seq{t}) \in A}
\bigvee_{i=1}^m \exists \seq{x_i}.
  \left(
    \phi_i \land \sem{\Gamma}{\mkk{A_i}{M \cup \set{\alpha}}{\circ}}{k}
  \right),
\end{align*}
where $\D(P)(\seq{t})
= \bigvee_{i=1}^m \exists \seq{x_i}. \left( \phi_i \land \bigwedge
A_i \right)$ and $\sembrack{\Gamma}$ is defined by:
\begin{align*}
\sembrack{\Gamma}(A') &= \bigwedge \bigcup_{(g, A, \phi, h) \in \Gamma} \sembrack{(g, A, \phi, h)}(A'), \\
\sembrack{(\bullet, A, \phi, h)}(A') &=
\set{\forall \seq{x}.\ \bigwedge A \land \phi \Rightarrow h \relmiddle{|} \set{\seq{x}}=\fvs{A,\phi,h}},
\end{align*}
\begin{align*}
\sembrack{(\alpha \triangleright P(\seq{t}), A, \phi, h)}(A') &= \\
&\!\!\!\!\!\!\!\!\!\!\!\!\!\!\!\!\!\!\!\!\!\!\!\!\!\!\!\!\!\!\!\!\!\!\!\!\!\!\!\!\!\!\!\!\!\!\!\!\!\!\!\!\!\!
\set{
\forall \seq{x}.\ \bigwedge A \land \phi \land \seq{t} = \seq{t}' \Rightarrow h
\relmiddle{|}
\begin{array}{l}
\mk{P}{M}(\seq{t}') \in A', \alpha \in M, \\
\set{\seq{x}}=\fvs{\seq{t},A,\phi,h}
\end{array}
}.
\end{align*}
Intuitively, $\sembrack{\Gamma}(A)$ represents the conjunction of
user-specified lemmas and induction hypotheses in $\Gamma$
instantiated for the atoms in $A$.
The soundness of the inductive proof system is now stated by the
following lemma (see \full{Appendix~\ref{soundness}}{an extended
version~\cite{}} for a proof).
\begin{lemma}[Soundness]
\label{lem:sound}
If \(\D; \Gamma; A; \phi \vdash h\) is derivable, then
there is $k$ such that
$\leastmodeld{\D} \models \sem{\Gamma}{A}{k} \land \bigwedge
A \land \phi \Rightarrow h$ holds.
\end{lemma}

The correctness of our Horn constraint solving method follows
immediately from Lemma~\ref{lem:sound} and
Corollary~\ref{cor:leastmodel} as follows.
\begin{theorem}
\label{cor:to_itp}
Suppose that $\h$ is an HCCS with $\goalof{\h} = \set{\bigwedge
A_i \land \phi_i \Rightarrow \bot}_{i=1}^m$.  It then follows that
$\rho \models \h$ for some $\rho$ if $\defof{\h}; \emptyset;
A_i; \phi_i \vdash
\bot$ is derivable for all $i=1,\dots,m$.
\end{theorem}
\begin{proof}
Suppose that $\defof{\h}; \emptyset; A_i; \phi_i \vdash \bot$ for all
$i=1,\dots,m$.  By Lemma~\ref{lem:sound} and the fact that
$\leastmodeld{\defof{\h}} \models \bigwedge
A_i \Rightarrow \sem{\emptyset}{A_i}{k}$, we get
$\leastmodeld{\defof{\h}} \models \bigwedge
A_i \land \phi_i \Rightarrow \bot$.
We therefore have $\leastmodeld{\defof{\h}} \models \goalof{\h}$.  It
then follows that $\rho \models \h$ for some $\rho$ by
Corollary~\ref{cor:leastmodel}.
\end{proof}

\subsection{Rule Application Strategy}
\label{sec:auto_induct}

We now elaborate on our rule application strategy shown in
Section~\ref{sec:overview}.  Because all the inference rules
except \rn{Valid$\bot$} and \rn{Valid$P$} add new knowledge to $A$
and/or $\phi$, we repeatedly apply them until \rn{Valid$\bot$}
and \rn{Valid$P$} close all the proof branches under consideration.
More specifically, we adopt the following strategy:
\begin{itemize}
\item Repeatedly apply the rules \rn{Apply$\bot$}, \rn{Apply$P$}, and \rn{Fold} if possible until no new knowledge is obtained.  (Even if the rules do not apply, applications of \rn{Induct} and \rn{Unfold} explained in the following items may make \rn{Apply$\bot$}, \rn{Apply$P$}, and \rn{Fold} applicable.)
\item If the current knowledge cannot be updated by using the above rules, select some atom from $A$ in a breadth-first manner, and apply the rule \rn{Induct} to the atom.
\item Apply the rule \rn{Unfold} whenever \rn{Induct} is applied.
\item Try to apply the rules \rn{Valid$\bot$} and \rn{Valid$P$} whenever the current knowledge is updated.
\end{itemize}

%% file: rules.tex
\begin{figure}[t]
\infrule{}{}
\vspace*{-2ex}
Perform induction on the derivation of the atom $P(\seq{t})$:
\infrule[Induct]
  {\mkk{P}{M}{\circ}(\seq{t}) \in A \andalso 
   \Gamma' = \Gamma \cup \set{(\alpha \triangleright P(\seq{t}),A,\phi,h)} \\
   \D; \Gamma'; (A \setminus \mkk{P}{M}{\circ}(\seq{t})) \cup \set{\mkk{P}{M}{\alpha}(\seq{t})}; \phi \vdash
   h \andalso (\alpha: \fresh)}
  {\D; \Gamma; A; \phi \vdash h}
  \rulesp
  
Case-analyze the last rule used (where $m$ rules are possible):
\infrule[Unfold]
  {\mkk{P}{M}{\alpha}(\seq{t}) \in A \andalso
   \D(P)(\seq{t}) = \bigvee_{i=1}^m \exists \seq{x_i}. \left( \phi_i \land \bigwedge A_i \right) \\
   \D; \Gamma; A \cup \mkk{A_i}{M \cup \set{\alpha}}{\circ}; \phi \land \phi_i
   \vdash h \andalso
   (\foreach\ i \in \set{1, \dots, m})}
  {\D; \Gamma; A; \phi \vdash h}
\rulesp

Apply an induction hypothesis or a user-specified lemma in $\Gamma$:
\infrule[Apply$\bot$]
  {(g,A',\phi',\bot) \in \Gamma \andalso
    \dom{\sigma} = \fvs{A'} \\
   \models \phi \Rightarrow \In{\sigma g}{A} \andalso
   \models \phi \Rightarrow \Sub{\sigma A'}{A}\\
   \set{\seq{x}}=\fvs{\phi'} \setminus \dom{\sigma} \andalso
   \D; \Gamma; A; \phi \land \forall \seq{x}.\neg(\sigma \phi') \vdash h}
  {\D; \Gamma; A; \phi \vdash h}
\rulesp

Apply an induction hypothesis or a user-specified lemma in $\Gamma$:
\infrule[Apply$P$]
  {(g,A',\phi',P(\seq{t})) \in \Gamma \andalso
    \dom{\sigma} = \fvs{A'} \cup \fvs{\seq{t}} \\
   \models \phi \Rightarrow \In{\sigma g}{A} \andalso
   \models \phi \Rightarrow \exists \seq{x}.(\sigma\phi') \andalso
   \models \phi \Rightarrow \Sub{\sigma A'}{A}\\
   \set{\seq{x}} = \fvs{\phi'} \setminus \dom{\sigma}
   \andalso
   \D; \Gamma; A \cup \set{\mkk{P}{\emptyset}{\circ}(\sigma
     \seq{t})}; \phi \vdash h}
  {\D; \Gamma; A; \phi \vdash h}
\rulesp

Apply a definite clause in $\D$:
\infrule[Fold]
  {(P(\seq{t}) \Leftarrow \phi' \land \bigwedge A') \in \D \andalso
    \dom{\sigma} = \fvs{A'} \cup \fvs{\seq{t}} \\
    \models \phi \Rightarrow \exists \seq{x}.(\sigma \phi') \andalso
    \models \phi \Rightarrow \Sub{\sigma A'}{A}\\
   \set{\seq{x}} = \fvs{\phi'} \setminus \dom{\sigma} \andalso
   \D; \Gamma; A \cup \set{\mkk{P}{\emptyset}{\circ}(\sigma
     \seq{t})}; \phi \vdash h}
  {\D; \Gamma; A; \phi \vdash h}
\rulesp

Check if the current knowledge entails the asserted proposition:

\rulesp
\begin{minipage}{0.5\hsize}
\infrule[Valid$\bot$]
  {\models \phi \Rightarrow \bot}
  {\D; \Gamma; A; \phi \vdash \bot}
\end{minipage}
\begin{minipage}{0.5\hsize}
\infrule[Valid$P$]
  {\models \phi \Rightarrow \In{P(\seq{t})}{A}}
  {\D; \Gamma; A; \phi \vdash P(\seq{t})}
\end{minipage}
\rulesp

Auxiliary functions:
\begin{align*}
\In{P(\seq{t})}{A} &\triangleq \bigvee_{P(\seq{t'}) \in A} \seq{t} = \seq{t}'\\
\In{\bullet}{A} &\triangleq \top\\
\In{\alpha \triangleright P(\seq{t})}{A} &\triangleq
\In{P(\seq{t})}{\set{\mk{P}{M}(\seq{t'}) \in A \mid \alpha \in M}}\\
\Sub{A_1}{A_2} &\triangleq \bigwedge_{P(\seq{t}) \in A_1} \In{P(\seq{t})}{A_2}
\end{align*}

\caption{The inference rules for the judgment $\D; \Gamma; A; \phi \vdash h$.}
\label{fig:rule}
\end{figure}

%% file: exp.tex
We have implemented a Horn constraint solver based on the proposed
method and integrated it, as a backend solver, with an existing
verification tool called Refinement
Caml~\cite{Unno2008,Unno2009,Unno2015}, a refinement type checking and
inference tool for the OCaml functional language based on Horn
constraint solving.
Our solver can generate a proof tree like the one in
Figure~\ref{fig:deriv} as a certificate, if the given Horn constraint
set is judged to have a solution.  Furthermore, our solver can
generate a counterexample, if the constraint set is judged to be
unsolvable.
We adopted Z3~\cite{Moura2008} as the underlying SMT solver.  The
details of the implementation are explained in Section~\ref{sec:impl}.
The web interface of the verification tool as well as all the
benchmark programs used in the experiments reported here are available
from \url{http://www.cs.tsukuba.ac.jp/~uhiro/}.

We have tested our constraint solver on two benchmark sets.  The first
set is 85 benchmarks from the test suite for automated induction
provided by the authors of the IsaPlanner system~\cite{Dixon2003}.
The benchmark set consists of verification problems of relational
specifications of pure mathematical functions on inductive data
structures, most of which cannot be verified by the previous Horn
constraint
solvers~\cite{Unno2009,Terauchi2010,Hoder2011,Gupta2011,Grebenshchikov2012,Rummer2013,McMillan2013,Unno2015}.
The benchmark set has also been used to evaluate previous automated
inductive theorem
provers~\cite{Leino2012,Sonnex2012,Claessen2013,Reynolds2015}.  The
experiment results on this benchmark set are reported in
Section~\ref{sec:isaplan}.

To demonstrate advantages of our novel combination of Horn constraint
solving with inductive theorem proving, we have prepared the second
benchmark set consisting of verification problems of (mostly
relational) specifications of programs that use various advanced
language features, which are naturally and automatically axiomatized
by our method using predicates defined by Horn clause constraints as
the least satisfying interpretation.  The experiment results on this
benchmark set are reported in Section~\ref{sec:various}.

\subsection{Implementation Details}
\label{sec:impl}

\input{impl}

\subsection{Experiments on IsaPlanner benchmark set}
\label{sec:isaplan}

The IsaPlanner benchmark set consists of 85 conjectures for total
recursive functions on inductively defined data structures such as
natural numbers, lists, and binary trees. We have translated these
conjectures into assertion safety verification problems of OCaml
programs.  In the translation, we encoded natural numbers using
integer primitives, and defined lists and binary trees as algebraic
data types in OCaml.  More specifically, natural numbers $Z$ and $S\
t$ are respectively encoded as $0$ and $t' + 1$ for $t'$ obtained by
encoding $t$.  To preserve the semantics of natural numbers, we
translated conjectures of the form $\forall x \in \natset.\ \phi$ into
$\forall x \in \intset.\left( x \geq 0 \Rightarrow \phi \right)$.

The translated verification problems are then verified by our
verification tool.
Our tool automatically reduced the verification problems into Horn
constraint solving problems by using the constraint generation
method~\cite{Unno2009}, and automatically (i.e., without using
user-specified lemmas) solved 68 out of 85 verification problems.  We
have manually analyzed the experiment results and found that 8 out of
17 failed verification problems require lemma discovery.  The other 9
problems caused timeout of Z3.  It was because the rule application
strategy implemented in our tool caused useless detours in proofs and
put heavier burden on Z3 than necessary.

The experiment results on the IsaPlanner benchmark set show that our
Horn-clause-based axiomatization of total recursive functions does not
cause significant negative impacts on the automation of induction;
According to \cite{Sonnex2012} that uses the IsaPlanner benchmark set
to compare state-of-the-art automated inductive theorem provers based
on logics of pure total functions over inductively-defined data
structures, IsaPlanner~\cite{Dixon2003} proved 47 out of 85,
Dafny~\cite{Leino2012} proved 45, ACL2s~\cite{Chamarthi2011} proved
74, and Zeno~\cite{Sonnex2012} proved 82.  The
HipSpec~\cite{Claessen2013} inductive prover and the SMT solver CVC4
extended with induction~\cite{Reynolds2015} are reported to have
proved 80.  In contrast to our Horn-clause-based method, these
inductive theorem provers can be, and in fact are directly applied to
prove the conjectures in the benchmark set, because the benchmark set
contains only pure total functions over inductively-defined data
structures.

It is also worth noting that, all the inductive provers that won best
results (greater than 70) on the benchmark set support automatic lemma
discovery, in a stark contrast to our tool.
For example, the above result (80 out of 85) of CVC4 is obtained when
they enable an automatic lemma discovery technique proposed
in \cite{Reynolds2015} and use a different encoding (called {\bf dti}
in \cite{Reynolds2015}) of natural numbers than ours.
When they disable the lemma discovery technique and use a similar
encoding to ours (called {\bf dtt} in \cite{Reynolds2015}), CVC4 is
reported to have proved 64.
Thus, we believe that extending our method with automatic lemma
discovery, which has been comprehensively studied by the automated
induction
community~\cite{Ireland1996,Kaufmann2000,Chamarthi2011,Sonnex2012,Claessen2013,Reynolds2015},
further makes induction-based Horn constraint solving powerful.

\input{exptable}

\subsection{Experiments on benchmark set consisting of programs with various advanced language features}
\label{sec:various}

We prepared and tested our tool with the second benchmark set
consisting of (mostly relational) assertion safety verification
problems of programs that use various advanced language features such
as partial (i.e., possibly non-terminating) functions, higher-order
functions, exceptions, non-determinism, algebraic data types, and
non-inductively defined data types (e.g., real numbers).  The
benchmark set also includes integer functions with complex recursion
and a verification problem concerning the equivalence of programs
written in different language paradigms.
All the verification problems except four (ID$19$--$22$ in
Table~\ref{tab:exp}) are relational ones where safe inductive
invariants are not expressible in \QFLIA, and therefore not solvable
by the previous Horn constraint solvers.
As shown in Section~\ref{sec:hornex}, these verification problems are
naturally and automatically axiomatized by our method using predicates
defined by Horn clause constraints as the least satisfying
interpretation.
By contrast, these assertion safety verification problems cannot be
straightforwardly axiomatized and proved by the previous automated
inductive theorem provers based on logics of pure total functions on
inductively-defined data structures: the axiomatization process of
these verification problems using pure total functions often requires
users' manual intervention and possibly causes a negative effect on
the automation of induction, because, in the process, one needs to
take into consideration the evaluation strategies and complex control
flows caused by higher-order functions and side-effects such as
non-termination, exceptions, and non-determinism.  Additionally, the
axiomatization process needs to preserve branching and calling context
information in order to perform path- and context-sensitive
verification.

Table~\ref{tab:exp} summarizes the experiment results on the benchmark
set.  The column ``specification'' represents the relational
specification verified and the column ``kind'' shows the kind of the
specification, where ``equiv'', ``assoc'', ``comm'', ``dist'', ``mono'',
``idem'', ``nonint'', and ``nonrel'' respectively represent the
equivalence, associativity, commutativity, distributivity,
monotonicity, idempotency, non-interference, and non-relational.
The column ``language features'' shows the language features used in
the verification problem, where each character has the following
meaning.
\begin{itemize}
\item[H:] higher-order functions
\item[E:] exceptions
\item[P:] partial (i.e., possibly non-terminating) functions
\item[D:] demonic non-determinism
\item[R:] real functions
\item[I:] integer functions with complex recursion
\item[N:] nonlinear functions
\item[C:] procedures written in different programming paradigms
\end{itemize}
The column ``result'' represents whether our verification method
succeeded \cmark or failed \xmark.  The column ``time'' represents the
elapsed time for verification in seconds.

Overall, the experiment results are promising, which show that our
tool can automatically solve relational verification problems that use
various advanced language features, in a practical time with
surprisingly few user-specified lemmas.
We also want to emphasize that the problem ID$5$, which required a
lemma, is a relational verification problem involving two function
calls with significantly different control flows: one recureses on $x$
and the other recurses on $y$.  Thus, the result demonstrates an
advantage of our induction-based method that it can exploit lemmas to
fill the gap between function calls with different control flows.
Our tool, however, failed to verify the distributivity ID$7$ of
$\mult$, the associativity ID$8$ of $\mult$, and the equivalence
ID$15$ of $\SUMFROMDOWN$ and $\SUMFROMUP$.
ID$7$ could be reduced to ID$6$ and solved, if a lemma $P_\mult(x, y,
r) \Rightarrow P_\mult(y, x, r)$, which represents the commutativity
of $\mult$, was used to rewrite the conjecture
\[P_\mult(x, y + z, s_1) \land P_\mult(x, y, s_2) \land P_\mult(x, z, s_3) \Rightarrow s_1=s_2+s_3\]
obtained from the specification $\mult\ x\ (y + z) = \mult\ x\ y
+ \mult\ x\ z$ into
\[P_\mult(y + z, x, s_1) \land P_\mult(y, x, s_2) \land P_\mult(z, x, s_3) \Rightarrow s_1=s_2+s_3\]
by replacing atoms of the form $P_\mult(t_1,t_2,t_3)$ with
$P_\mult(t_2,t_1,t_3)$.  The rule \rn{ApplyP}, however, replaces each
atom $P_\mult(t_1, t_2, t_3)$ with $P_\mult(t_1, t_2, t_3) \land
P_\mult(t_2, t_1, t_3)$ instead by keeping the original atom so that
we can monotonically increase the current knowledge.  Our tool
supports an option for the rule \rn{ApplyP} of eliminating the
original atom, and if it is enabled, ID$7$ is verified.
The associativity verification problem ID$8$ is even more difficult.
In addition to the above lemma, a lemma $P_\mult(x + y, z,
r) \Rightarrow \exists s_1,s_2.( P_\mult(x, z, s_1) \land P_\mult(y,
z, s_2) \land r = s_1 + s_2 )$ is required.  This lemma, however, is
currently not of the form supported by our inductive proof system.
In ID$15$, the functions $\SUMFROMDOWN$ and $\SUMFROMUP$ use different
recursion parameters (resp. $y$ and $x$),
and requires lemmas $P_\SUMFROMDOWN(x, y, s) \Rightarrow \exists
s_1,s_2.( P_\SUMFROMDOWN(0, y, s_1) \land P_\SUMFROMDOWN(0, x - 1,
s_2) \land s = s_1 - s_2 )$ and $P_\SUMFROMUP(x, y,
s) \Rightarrow \exists s_1,s_2.( P_\SUMFROMDOWN(0, y, s_1) \land
P_\SUMFROMDOWN(0, x - 1, s_2) \land s = s_1 - s_2 )$.  These lemmas
are provable by induction on the derivation of $P_\SUMFROMDOWN(x, y,
s)$ and $P_\SUMFROMUP(x, y, s)$, respectively.  However, as in the
case of ID$8$, our proof system does not support the form of the
lemmas.  To put it differently, ID$8$ and ID$15$ demonstrate the
incompleteness of our inductive proof system.  Our future work thus
includes an extension of the proof system to support more general form
of lemmas and judgments.

%% file: impl.tex
This section describes details of the implementation.  We explain how
to check the correctness of user specified lemmas and how to generate
a counterexample if the given Horn constraint set has no solution,
respectively in Sections~\ref{sec:lemma} and \ref{sec:refute}.
Section~\ref{sec:assign} describes implementation details of the
rules \rn{Apply$\bot$}, \rn{Apply$P$}, and \rn{Fold} in
Figure~\ref{fig:rule}.  In particular, we discuss how to find an
assignment $\sigma$ for free term variables that occur in the element
of $\Gamma$ selected by the rules.

\subsubsection{Checking the correctness of user-specified lemmas}
\label{sec:lemma}

Our system allows users to specify lemmas as the initial $\Gamma$.
Our tool checks that $\D; \emptyset; A; \phi \vdash h$ is derivable
for each user-specified lemma $(\bullet,A,\phi,h)$ by using the exact
same rules in Figure~\ref{fig:rule}.
We use the rule \rn{Fold}, in addition to the rules \rn{Apply$\bot$}
and \rn{Apply$P$}, to update the current knowledge.  To avoid
redundant applications of \rn{Fold}, we select only definite clauses
in $\D$ with the head of the form $P(\seq{t}')$ if $h=P(\seq{t})$ and
we do not use \rn{Fold} at all if $h=\bot$.

\subsubsection{Counterexample generation}
\label{sec:refute}

Our tool can conclude that the goal clause (or the user-specified
lemma) currently solving has no solution if a subgoal of the form
$\D; \Gamma; A; \phi \vdash \bot$ satisfying the following conditions
is obtained:
\begin{itemize}
\item all the atoms in $A$ are already unfolded by the rule \rn{Unfold} but
\item $\phi$ is satisfiable.
\end{itemize}
Note that the first condition ensures that the $\phi$-part of the
current knowledge under-approximates the body of the goal clause.

Our tool then returns a satisfying model of $\phi$ found by the
underlying SMT solver as a counterexample witnessing the unsolvability
of the given Horn constraint set.  Some readers may notice that the
counterexample generation is essentially the same as the execution of
constraint logic programs~\cite{Jaffar1994}.

\subsubsection{Finding an assignment $\sigma$ for quantifier instantiation}
\label{sec:assign}

We here explain how to find $\sigma$ for instantiating quantified
variables of lemmas and induction hypotheses by using an SMT solver in
the implementation of the rule \rn{Apply$\bot$}.  The same technique
is also used for finding $\sigma$ in the rules \rn{Fold}
and \rn{Apply$P$}.

Recall that, in order to apply the rule \rn{Apply$\bot$} to a judgment
$\D; \Gamma; A; \phi \vdash h$, we need to find an assignment $\sigma$
for free term variables in $(g, A', \phi', h') \in \Gamma$.  Note here
that the atom that occurs in $g$ also occurs in $A'$.
We below assume that all the arguments of the atoms in $A'$ are
distinct term variables.
This does not lose generality because we can always replace
$P(\seq{t}) \in A'$ by $P(\seq{x})$ with fresh $\seq{x}$ by adding the
constraint $\seq{x} = \seq{t}$ to $\phi'$.  First of all, we
construct, for each $P(\seq{x}) \in A'$, the set
$\set{\set{\seq{x} \mapsto \seq{t}_i}}_{i=1}^m$ of assignments, where
$P(\seq{t}_1),\dots,P(\seq{t}_m) \in A$.
Here, if $g=\alpha \triangleright P(\seq{x})$, the set
$\set{\set{\seq{x} \mapsto \seq{t}_i}}_{i=1}^m$ of assignments is
constructed only from $P^{M_1}(\seq{t}_1),\dots,P^{M_m}(\seq{t}_m) \in
A$ such that $\alpha \in M_i$.
For example, let us consider $g=\alpha \triangleright P_1(\seq{x}_1)$,
$A' =\set{P_1(\seq{x}_1), P_1(\seq{x}_2), P_2(\seq{x}_3)}$ and $A
= \set{\mk{P_1}{\set{\alpha}}(\seq{t}_1), \mk{P_1}{\emptyset}(\seq{t}_2), \mk{P_2}{\emptyset}(\seq{t}_3)}$.
For the atoms $P_1(\seq{x}_1)$, $P_1(\seq{x}_2)$, and
$P_2(\seq{x}_3)$, we respectively obtain the sets
$\set{\set{\seq{x}_1 \mapsto \seq{t}_1}}$,
$\{\{\seq{x}_2 \mapsto \seq{t}_1\}, \{\seq{x}_2 \mapsto \seq{t}_2\}\}$,
and $\set{\set{\seq{x}_3 \mapsto \seq{t}_3}}$ of assignments.
We then compute all the combination of assignments, and filter out
those that contradict with $\phi$.  For the above example, we obtain
the following two as candidates of $\sigma$:
\begin{align*}
\sigma_1 = \set{\seq{x}_1 \mapsto \seq{t}_1,\seq{x}_2 \mapsto \seq{t}_1,\seq{x}_3 \mapsto \seq{t}_3} \\
\sigma_2 = \set{\seq{x}_1 \mapsto \seq{t}_1,\seq{x}_2 \mapsto \seq{t}_2,\seq{x}_3 \mapsto \seq{t}_3}
\end{align*}

For the rules \rn{Fold} and \rn{Apply$P$}, we use the same technique
explained above for \rn{Apply$\bot$}, but additionally check the
condition $\models \phi \Rightarrow \exists \seq{x}.\ \sigma \phi'$,
where $\set{\seq{x}} = \fvs{\phi'} \setminus \dom{\sigma}$.

%% file: exptable.tex
\begin{table*}[t]
\caption{Experiment results on programs that use various language features}
\label{tab:exp}
\begin{minipage}{\textwidth}
  \begin{center}
  \begin{threeparttable}
    \begin{tabular}{|c|l|l|c|c|r|} \hline
      ID & specification & kind & language features & result & time (sec.) \\
      \hline \hline
      1 & $\mult\ x\ y + a = \multacc\ x\ y\ a$ & equiv & P &\cmark & 0.257 \\
      \hline
      2 & $\mult\ x\ y = \multacc\ x\ y\ 0$ & equiv & P & \cmark\tnote{\dag} & 0.435 \\
      \hline
      3 & $\mult\ (1+x)\ y = y+\mult\ x\ y$ & equiv & P & \cmark & 0.233 \\
      \hline
      4 & $y \geq 0 \Rightarrow \mult\ x\ (1+y) = x+\mult\ x\ y$ & equiv & P & \cmark & 0.248 \\
      \hline
      5 & $\mult\ x\ y = \mult\ y\ x$ & comm & P & \cmark\tnote{\ddag} & 0.345 \\
      \hline
      6 & $\mult\ (x + y)\ z = \mult\ x\ z + \mult\ y\ z$ & dist & P & \cmark & 1.276 \\
      \hline
      7 & $\mult\ x\ (y + z) = \mult\ x\ y + \mult\ x\ z$ & dist & P & \xmark & n/a \\
      \hline
      8 & $\mult\ (\mult\ x\ y)\ z = \mult\ x\ (\mult\ y\ z)$ & assoc & P & \xmark & n/a \\
      \hline
      9 & $0 \leq x_1 \leq x_2 \land 0 \leq y_1 \leq y_2 \Rightarrow \mult\ x_1\ y_1 \leq \mult\ x_2\ y_2$ & mono & P & \cmark & 0.265 \\
      \hline
      \hline
      10 & $\SUM\ x + a = \SUMACC\ x\ a$ & equiv & & \cmark & 0.384 \\
      \hline
      11 & $\SUM\ x = x + \SUM\ (x-1)$ & equiv & & \cmark & 0.272 \\
      \hline
      12 & $x \leq y \Rightarrow \SUM\ x \leq \SUM\ y$ & mono & & \cmark & 0.350 \\
      \hline
      13 & $x \geq 0 \Rightarrow \SUM\ x = \SUMFROMDOWN\ 0\ x$ & equiv & P & \cmark & 0.312 \\
      \hline
      14 & $x < 0 \Rightarrow \SUM\ x = \SUMFROMUP\ x\ 0$ & equiv & P & \cmark & 0.368 \\
      \hline
      15 & $\SUMFROMDOWN\ x\ y = \SUMFROMUP\ x\ y$ & equiv & P & \xmark & n/a \\
      \hline
      \hline
      16 & $\SUM\ x = \APPLY\ \SUM\ x$ & equiv & H & \cmark & 0.286 \\
      \hline
      17 & $\mult\ x\ y = \APPLY2\ \mult\ x\ y$ & equiv & H, P & \cmark & 0.279 \\
      \hline
      18 & $\REPEAT\ x\ (\ADD\ x)\ a\ y\ = a + \MULT\ x\ y$ & equiv & H, P & \cmark & 0.317 \\
      \hline
      \hline
      19 & $x \leq 101 \Rightarrow \MCC\ x = 91$ & nonrel & I & \cmark & 0.165 \\
      \hline
      20 & $x \geq 0 \land y \geq 0 \Rightarrow \ACK\ x\ y > y$ & nonrel & I & \cmark & 0.212 \\
      \hline
      21 & $x \geq 0 \Rightarrow 2 \times \SUM\ x = x \times (x + 1)$ & nonrel & N & \cmark & 0.196 \\
      \hline
      22 & $\mathtt{dyn\_sys}\ 0. \notreds \ASSERT\ \FALSE$ & nonrel & R,N & \cmark & 0.144 \\
      \hline
      \hline
      23 & $\FMOD\ y\ x = \FMOD\ y\ (\FMOD\ y\ x)$ & idem & P & \cmark & 7.712 \\
      \hline
      24 & $\NONINTER\ h_1\ l_1\ l_2\ l_3 = \NONINTER\ h_2\ l_1\ l_2\ l_3$ & nonint & P & \cmark & 0.662 \\
      \hline
      25 & $\mathbf{try}\ \FINDOPT\ p\ l = \SOME\ (\FIND\ p\ l)\ \mathbf{with}\ $ & & & & \\
         & $\quad \NOTFOUND \rightarrow \FINDOPT\ p\ l = \NONE$ & equiv & H, E & \cmark & 0.758 \\
      \hline
      26 & $\mathbf{try}\ \MEM\ (\FIND\ ((=)\ x)\ l)\ l\ \mathbf{with}\ \NOTFOUND \rightarrow \lnot(\MEM\ x\ l)$ & equiv & H, E & \cmark & 0.764 \\
      \hline
      27 & $\SUML\ l = \FOLDL\ (+)\ 0\ l$ & equiv & H & \cmark & 3.681 \\
      \hline
      28 & $\SUML\ l = \FOLDR\ (+)\ l\ 0$ & equiv & H & \cmark & 0.329 \\
      \hline
      29 & $\mathtt{sum\_fun}\ \mathtt{randpos}\ n > 0$ & equiv & H,D & \cmark & 0.240 \\
      \hline
      \hline
      30 & $\MULT\ x\ y = \MULTC(x, y)$ & equiv & P, C & \cmark & 0.217 \\
      \hline
    \end{tabular}
\begin{tablenotes}
\item[\dag] A lemma $P_\multacc(x, y, a, r) \Rightarrow P_\multacc(x, y, a - x, r - x)$ is used
\item[\ddag] A lemma $P_\mult(x, y, r) \Rightarrow P_\mult(x - 1, y, r - y)$ is used
\end{tablenotes}
$P_f$ above represents the predicate that axiomatizes the function $f$. \\
The experiments were conducted on a machine with Intel(R) Xeon(R) CPU E5-2680 v3 (2.50 GHz, 16 GB of memory).
\end{threeparttable}
  \end{center}
\end{minipage}
\end{table*}

%% file: related.tex
As discussed in Section~\ref{sec:intro}, Horn constraint solving
methods have been extensively
studied~\cite{Unno2009,Terauchi2010,Hoder2011,Gupta2011,Grebenshchikov2012,Rummer2013,McMillan2013,Unno2015}.
In contrast to the proposed induction based method, these methods do
not support Horn clause constraints over the theories of algebraic
data structures and nonlinear arithmetics, and cannot verify most if
not all relational specifications shown in Section~\ref{sec:exp}.

Because state-of-the-art SMT solvers such as Z3~\cite{Moura2008} and
CVC4 support quantifier instantiation heuristics, one may think that
they alone are sufficient for checking the validity of the logical
interpretation of Horn clause constraints shown in
Section~\ref{sec:notation}.  However, they alone are not sufficient
for proving most conjectures that require nontrivial use of induction
such as the benchmark problems in Section~\ref{sec:exp}.\footnote{This
point is also mentioned in the tutorial of Z3
(\url{http://rise4fun.com/Z3/tutorial/guide}).}  In
fact, \cite{Reynolds2015} reports that Z3 (resp. CVC4 without
induction) alone have proved only 35 (resp. 34) out of 85 problems in
the IsaPlanner benchmark set.

Automated inductive theorem proving techniques and tools have long
been studied, for example and to name a few: the Boyer-Moore theorem
provers~\cite{Kaufmann2000} like ACL2s~\cite{Chamarthi2011}, rewriting
induction provers~\cite{Reddy1990} like SPIKE~\cite{Bouhoula1992},
proof planners like
CLAM~\cite{Bundy1990,Bundy2001,Ireland1996,Johansson2010} and
IsaPlanner~\cite{Dixon2003,Dixon2004}, and SMT-based induction provers
like Leon~\cite{Suter2011a}, Dafny~\cite{Leino2012},
Zeno~\cite{Sonnex2012}, HipSpec~\cite{Claessen2013}, and
CVC4~\cite{Reynolds2015}.
These automated provers are mostly based on logics of pure total
functions over inductive data types.  Consequently, users of these
provers are required to axiomatize advanced language features and
specifications (e.g., ones discussed in Section~\ref{sec:hornex})
using pure total functions as necessary.
The axiomatization process, however, is non-trivial, error-prone, and
possibly causes a negative effect on the automation of induction.  For
example, if a partial function (e.g., $f(x)=f(x)+1$) is input, Zeno
goes into an infinite loop and CVC4 is unsound (unless control
literals proposed in \cite{Suter2011a} are used in the
axiomatization).  We have also confirmed that CVC4 failed to verify
complex integer functions like the McCarthy 91 and the Ackermann
functions (resp. ID$19$ and ID$20$ in Table~\ref{tab:exp}).
By contrast, our method supports advanced language features and
specifications via Horn-clause encoding of their semantics based on
program logics such as Hoare logics and refinement type systems.

To aid verification of relational specifications of functional
programs, Giesl~\cite{Giesl2000} proposed context-moving
transformations and Asada et al.~\cite{Asada2015} proposed a kind of
tupling transformation.  SymDiff~\cite{Lahiri2012,Hawblitzel2013} is a
transformation-based tool built on top of Boogie~\cite{Barnett2006}
for equivalence verification of imperative programs.
Self-composition~\cite{Barthe2004} is a program transformation
technique to reduce k-safety~\cite{Terauchi2005,Clarkson2008}
verification into ordinary safety verification, and has been applied
to non-interference~\cite{Terauchi2005,Unno2006,Barthe2011} and
regression verification~\cite{Felsing2014} of imperative programs.
These transformations are useful for some patterns of relational
verification problems, which are, however, less flexible than our
approach based on a more general principle of induction.  For example,
Asada et al.'s transformation enables verification of the functional
equivalence of recursive functions with the same recursion pattern
(e.g., ID$1$ in Table~\ref{tab:exp}), but does not help verification
of the commutativity of $\mult$ (ID$5$ in Table~\ref{tab:exp}).
Because each transformation is designed for a particular target
language, the transformations cannot be applied to aid relational
verification across programs written in different paradigms (e.g.,
ID$30$ in Table~\ref{tab:exp}).  Moreover, the correctness proof of
the transformations tends to be harder because it involves the
operational semantics of the target language, which is complex
compared to the logical semantics of Horn clause constraints.

There have been proposed program logics for relational
verification~\cite{Ciobaca2014,Barthe2012,Godlin2013,Barthe2015}.  In
particular, the relational refinement type system proposed
in \cite{Barthe2015} can be applied to differential privacy and other
relational security verification problems of higher-order functional
programs. This approach is, however, not automated.

%% file: conc.tex
We have proposed a novel Horn constraint solving method based on an
inductive proof system and an SMT-based technique to automate proof
search in the system.  We have shown that our method is able to solve
Horn clause constraints obtained from relational verification problems
that were not possible with the previous methods based on
interpolating theorem proving.  Furthermore, our novel combination of
Horn clause constraints with inductive theorem proving enabled our
method to automatically axiomatize and verify relational
specifications of programs that use various advanced language
features.

As a future work, we are planning to extend our inductive proof system
to support more general form of lemmas and judgments.  We are also
planning to extend our proof search method to support automatic lemma
discovery as in the state-of-the-art inductive theorem
provers~\cite{Chamarthi2011,Sonnex2012,Claessen2013,Reynolds2015}.  To
aid users to better understand verification results of our method, it
is important to generate a symbolic representation of a solution of
the original Horn constraint set from the found inductive proof.  It
is however often the case that a solution of Horn constraint sets that
require relational analysis (e.g., $\hmult$) is not expressible by a
formula of the underlying logic.  It therefore seems fruitful to
generate a symbolic representation of mutual summaries in the sense
of \cite{Hawblitzel2013} across multiple predicates (e.g., $P,Q$ of
$\hmult$).

%% file: soundness.tex
\input{lemmas}

We now prove Lemma~\ref{lem:sound}.
\begin{proof}[Proof of Lemma~\ref{lem:sound}]
By induction on the derivation of $\D; \Gamma; A; \phi \vdash h$.
\begin{itemize}
\item {\bf Case \rn{Induct}:}
We have
\begin{itemize}
\item $\mkk{P}{M}{\circ}(\seq{t}) \in A$,
\item $\Gamma' = \Gamma \cup \set{(\alpha \triangleright P(\seq{t}),A,\phi,h)}$,
\item $\D; \Gamma'; A'; \phi \vdash h$,
\item $A' = (A \setminus \mkk{P}{M}{\circ}(\seq{t})) \cup \set{\mkk{P}{M}{\alpha}(\seq{t})}$, and
\item $\alpha$ is fresh.
\end{itemize}
By I.H., there is $k'$ such that
\[
\leastmodeld{\D} \models \sem{\Gamma'}{A'}{k'} \land \bigwedge A' \land \phi \Rightarrow h
\]
By Lemma~\ref{lem:gamma}, we get
\begin{align*}
\leastmodeld{\D} \models
& \sem{\set{(\alpha \triangleright
    P(\seq{t}),A,\phi,h)}}{\set{\mkk{P}{M}{\alpha}(\seq{t})}}{k'} \land \\
& \bigwedge A' \land \phi \land \sem{\Gamma}{A}{k'} \Rightarrow h
\end{align*}
By the fact $\models A \Leftrightarrow A'$ and the definition of
$\sem{\bullet}{\bullet}{\bullet}$, we get
\begin{align*}
\leastmodeld{\D} \models
& \sem{\set{(\alpha \triangleright
    P(\seq{t}),A,\phi \land \sem{\Gamma}{A}{k'},h)}}{\set{\mkk{P}{M}{\alpha}(\seq{t})}}{k'} \land \\
& \bigwedge A \land \phi \land \sem{\Gamma}{A}{k'} \Rightarrow h
\end{align*}
Therefore, by Lemma~\ref{lem:min2}, for $k=k'$, we get
\[
\leastmodeld{\D} \models \sem{\Gamma}{A}{k} \land \bigwedge A \land \phi \Rightarrow h
\]

\item {\bf Case \rn{Unfold}:}
For each $i \in \set{1, \dots, m}$, we have
\begin{itemize}
\item $\mkk{P}{M}{\alpha}(\seq{t}) \in A$,
\item $\D(P)(\seq{t}) = \bigvee_{i=1}^m \exists \seq{x_i}.\left( \phi_i \land \bigwedge A_i \right)$, and
\item $\D; \Gamma; A \cup \mkk{A_i}{M \cup \set{\alpha}}{\circ}; \phi \land \phi_i \vdash h$.
\end{itemize}
By I.H., for each $i \in \set{1,\dots,m}$, there is $k_i$ such that 
\[
\leastmodeld{\D} \models \sem{\Gamma}{A \cup \mkk{A_i}{M \cup \set{\alpha}}{\circ}}{k_i} \land \bigwedge \left(A \cup \mkk{A_i}{M \cup \set{\alpha}}{\circ} \right) \land \phi \land \phi_i \Rightarrow h
\]
It then follows immediately that
\begin{align*}
\leastmodeld{\D} \models
& \bigwedge A \land \phi \land \\
& \bigvee_{i=1}^m \left(\phi_i \land \sem{\Gamma}{A \cup \mkk{A_i}{M \cup \set{\alpha}}{\circ}}{k_i} \land \bigwedge \mkk{A_i}{M \cup \set{\alpha}}{\circ} \right) \Rightarrow h
\end{align*}
By Lemma~\ref{lem:unfold3}, we get
\[
\leastmodeld{\D} \models \sem{\Gamma}{A}{k'+1} \Rightarrow \bigvee_{i=1}^m \exists \seq{x_i}.\left( \phi_i \land \sem{\Gamma}{A \cup \mkk{A_i}{M \cup \set{\alpha}}{\circ}}{k'} \right),
\]
where $k'=\max \set{k_i}_{i=1}^m$.  Therefore, by
Lemma~\ref{lem:unfold1}, we obtain
\[
\leastmodeld{\D} \models \sem{\Gamma}{A}{k'+1} \land \bigwedge A \land
\phi \land \left(\bigvee_{i=1}^m \left( \phi_i \land \bigwedge A_i \right)\right) \Rightarrow h
\]
From the facts $\leastmodeld{\D} \models P(\seq{t}) \Leftrightarrow
\bigvee_{i=1}^m \exists \seq{x_i}.\left( \phi_i \land \bigwedge A_i
\right)$ and $\mkk{P}{M}{\alpha}(\seq{t}) \in A$, for $k=k'+1$, we get
\[
\leastmodeld{\D} \models \sem{\Gamma}{A}{k} \land \bigwedge A \land \phi \Rightarrow h.
\]

\item {\bf Case \rn{Apply$\bot$}:}
We have
\begin{itemize}
\item $(g,A',\phi',\bot) \in \Gamma$,
\item $\dom{\sigma} = \fvs{A'}$,
\item $\models \phi \Rightarrow \In{\sigma g}{A}$,
\item $\models \phi \Rightarrow \Sub{\sigma A'}{A}$,
\item $\set{\seq{x}}=\fvs{\phi'} \setminus \dom{\sigma}$, and
\item $\D; \Gamma; A; \phi \land \forall \seq{x}.\neg (\sigma \phi') \vdash h$.
\end{itemize}
By $(g, A', \phi', \bot) \in \Gamma$,
$\models \phi \Rightarrow \In{\sigma g}{A}$, and
Lemmas~\ref{lem:subder} and \ref{lem:unfold1}, for some $k_1$, we get
\[
\models \sem{\Gamma}{A}{k_1} \land \bigwedge \sigma A' \land \phi
\Rightarrow \forall \seq{x}. \neg (\sigma \phi')
\]
By $\models \phi \Rightarrow \Sub{\sigma A'}{A}$ and
Lemma \ref{lem:sub}, we obtain
\[
\models \sem{\Gamma}{A}{k_1} \land \bigwedge A \land \phi
\Rightarrow \forall \seq{x}. \neg (\sigma \phi')
\]
By I.H., there is $k_2$ such that
\[
\leastmodeld{\D} \models \sem{\Gamma}{A}{k_2} \land \bigwedge A \land \phi \land \forall \seq{x}. \neg(\sigma \phi')\Rightarrow h
\]
Therefore, for $k=\max(k_1,k_2)$, by Lemma~\ref{lem:unfold1}, we obtain
\[
\leastmodeld{\D} \models \sem{\Gamma}{A}{k} \land \bigwedge A \land \phi \Rightarrow h.
\]

\item {\bf Case \rn{Apply$P$}:}
We have
\begin{itemize}
\item $(g,A',\phi',P(\seq{t})) \in \Gamma$,
\item $\dom{\sigma} = \fvs{A'} \cup \fvs{\seq{t}}$,
\item $\models \phi \Rightarrow \In{\sigma g}{A}$,
\item $\models \phi \Rightarrow \exists \seq{x}.(\sigma\phi')$,
\item $\models \phi \Rightarrow \Sub{\sigma A'}{A}$,
\item $\set{\seq{x}} = \fvs{\phi'} \setminus \dom{\sigma}$,
\item $\D; \Gamma; A \cup \set{\mkk{P}{\emptyset}{\circ}(\sigma \seq{t})}; \phi \vdash h$
\end{itemize}
By $(g, A', \phi', P(\seq{t})) \in \Gamma$,
$\models \phi \Rightarrow \In{\sigma g}{A}$, and
Lemmas~\ref{lem:subder} and \ref{lem:unfold1}, for some $k_1$, we get
\[
\models \sem{\Gamma}{A}{k_1} \land \bigwedge \sigma A' \land \phi
\land \exists \seq{x}.(\sigma \phi') \Rightarrow P(\sigma \seq{t})
\]
Then, by $\models \phi \Rightarrow \Sub{\sigma A'}{A}$,
Lemma~\ref{lem:sub}, and
$\models \phi \Rightarrow \exists \seq{x}.(\sigma\phi')$, we get
\[
\models \sem{\Gamma}{A}{k_1} \land \bigwedge A \land \phi \Rightarrow  P(\sigma \seq{t})
\]
By I.H., there is $k_2$ such that
\[
\leastmodeld{\D} \models \sem{\Gamma}{A \cup \set{\mkk{P}{\emptyset}{\circ}(\sigma \seq{t})}}{k_2} \land \bigwedge
\left(A \cup \set{\mkk{P}{\emptyset}{\circ}(\sigma \seq{t})}\right)
\land \phi \Rightarrow h
\]
By Lemma~\ref{lem:unfold2}, we obtain
\[
\leastmodeld{\D} \models
\sem{\Gamma}{A}{k_2} \land P(\sigma \seq{t})
\Rightarrow \sem{\Gamma}{A \cup \set{\mkk{P}{\emptyset}{\circ}(\sigma \seq{t})}}{k_2}
\]
It then follows that
\[
\leastmodeld{\D} \models \sem{\Gamma}{A}{k_2} \land \bigwedge
\left(A \cup \set{\mkk{P}{\emptyset}{\circ}(\sigma \seq{t})}\right)
\land \phi \Rightarrow h
\]
Therefore, for $k=\max(k_1,k_2)$, by Lemma~\ref{lem:unfold1}, we get
\[
\leastmodeld{\D} \models \sem{\Gamma}{A}{k} \land \bigwedge A \land \phi \Rightarrow h
\]

\item {\bf Case \rn{Fold}:}
We have
\begin{itemize}
\item $\left(P(\seq{t}) \Leftarrow \phi' \land \bigwedge A'\right) \in \D$,
\item $\dom{\sigma} = \fvs{A'} \cup \fvs{\seq{t}}$,
\item $\models \phi \Rightarrow \exists \seq{x}.(\sigma \phi')$,
\item $\models \phi \Rightarrow \Sub{\sigma A'}{A}$,
\item $\set{\seq{x}} = \fvs{\phi'} \setminus \dom{\sigma}$,
\item $\D; \Gamma; A \cup \set{\mkk{P}{\emptyset}{\circ}(\sigma \seq{t})}; \phi \vdash h$
\end{itemize}
By $P(\seq{t}) \Leftarrow \phi' \land \bigwedge A' \in \D$,
$\dom{\sigma} = \fvs{A'} \cup \fvs{\seq{t}}$, and $\set{\seq{x}}
= \fvs{\phi'} \setminus \dom{\sigma}$, we obtain
\[
\leastmodeld{\D} \models \bigwedge \sigma A' \land \exists
\seq{x}.(\sigma \phi') \Rightarrow P(\sigma \seq{t})
\]
By $\models \phi \Rightarrow \Sub{\sigma A'}{A}$, Lemma~\ref{lem:sub},
and $\models \phi \Rightarrow \exists \seq{x}.(\sigma \phi')$, we get
\[
\leastmodeld{\D} \models \bigwedge A \land \phi \Rightarrow P(\sigma \seq{t})
\]
By I.H., there is $k'$ such that
\[
\leastmodeld{\D} \models \sem{\Gamma}{A \cup
  \set{\mkk{P}{\emptyset}{\circ}(\sigma \seq{t})}}{k'} \land \bigwedge
\left(A \cup \set{\mkk{P}{\emptyset}{\circ}(\sigma \seq{t})}\right)
\land \phi \Rightarrow h
\]
By Lemma~\ref{lem:unfold2}, we obtain
\[
\leastmodeld{\D} \models
\sem{\Gamma}{A}{k'} \land P(\sigma \seq{t})
\Rightarrow \sem{\Gamma}{A \cup \set{\mkk{P}{\emptyset}{\circ}(\sigma \seq{t})}}{k'}
\]
It then follows that
\[
\leastmodeld{\D} \models \sem{\Gamma}{A}{k'} \land \bigwedge
\left(A \cup \set{\mkk{P}{\emptyset}{\circ}(\sigma \seq{t})}\right)
\land \phi \Rightarrow h
\]
Therefore, for $k=k'$, we obtain
\[
\leastmodeld{\D} \models \sem{\Gamma}{A}{k} \land \bigwedge A \land \phi \Rightarrow h
\]

\item {\bf Case \rn{Valid$\bot$}:} We have $h=\bot$ and $\models \phi
  \Rightarrow \bot$.  Therefore, $\leastmodeld{\D} \models
  \sem{\Gamma}{A}{k} \land \bigwedge A \land \phi \Rightarrow h$ holds
  for any $k$.

\item {\bf Case \rn{Valid$P$}:} We have $h=P(\seq{t})$ and $\models
  \phi \Rightarrow \In{P(\seq{t})}{A}$.  By lemma~\ref{lem:mem}, we
  get $\models \bigwedge A \land \phi \Rightarrow P(\seq{t})$.  It
  then follows that
  $\leastmodeld{\D} \models \sem{\Gamma}{A}{k} \land \bigwedge
  A \land \phi \Rightarrow P(\seq{t})$ for any $k$.

\end{itemize}
\end{proof}

%% file: lemmas.tex
We first show lemmas used to prove Lemma~\ref{lem:sound}.
\begin{lemma}
\label{lem:mem}
$\models \phi \Rightarrow \In{P(\seq{t})}{A}$ implies
$\models \bigwedge A \land \phi \Rightarrow P(\seq{t})$.
\end{lemma}
\begin{proof}
By the definition of $\In{P(\seq{t})}{A}$.
\end{proof}

\begin{lemma}
\label{lem:sub}
$\models \phi \Rightarrow \Sub{A_1}{A_2}$ implies $\models \bigwedge
A_2 \land \phi \Rightarrow \bigwedge A_1$.
\end{lemma}
\begin{proof}
By the definition of $\Sub{A_1}{A_2}$ and Lemma~\ref{lem:mem}.
\end{proof}

\begin{lemma}
\label{lem:subder}
Suppose that $(g, A', \phi', h) \in \Gamma$ and
$\models \phi \Rightarrow \In{\sigma g}{A}$ for some $\sigma$ with
$\dom{\sigma} = \fvs{A'} \cup \fvs{h}$.  It then follows that
$\models \sembrack{\Gamma}(A) \land \bigwedge \sigma
A' \land \phi \Rightarrow (\sigma h \lor \forall \seq{x}. \neg
(\sigma \phi'))$, where $\set{\seq{x}}
= \fvs{\phi'} \setminus \dom{\sigma}$.
\end{lemma}
\begin{proof}
We perform a case analysis on $g$.
\begin{itemize}
\item Suppose that $g = \bullet$.  By the definition of $\sembrack{\Gamma}(A)$, we
obtain
\[
\models \sembrack{\Gamma}(A) \Rightarrow \forall \seq{x}'.\left(\bigwedge A'\land \phi' \Rightarrow h \right),
\]
where $\set{\seq{x}'} = \fvs{A'} \cup \fvs{\phi'} \cup \fvs{h}$.
Therefore, we get
\[
\models \sembrack{\Gamma}(A) \Rightarrow \bigwedge \sigma A' \land \exists \seq{x}.(\sigma \phi') \Rightarrow \sigma h,
\]
where $\set{\seq{x}} = \fvs{\phi'} \setminus \dom{\sigma}$.  It
immediately follows that
\[
\models \sembrack{\Gamma}(A) \land \bigwedge \sigma A' \land \phi \Rightarrow
(\sigma h \lor \forall \seq{x}. \neg (\sigma \phi')).
\]
\item Suppose that $g = \alpha \triangleright P(\seq{t})$.  We assume that $\phi$ is not equivalent to $\bot$ (otherwise, the lemma is trivial).  By $\models \phi \Rightarrow \In{\sigma g}{A}$, there is $\mkk{P}{M}{}(\seq{t}') \in A$ such that $\alpha \in M$ and
$\models \phi \Rightarrow \sigma \seq{t}=\seq{t}'$.  By the definition
of $\sembrack{\Gamma}(A)$, we obtain
\[
\models \sembrack{\Gamma}(A) \Rightarrow \forall \seq{x}'.\left(\bigwedge A'\land \phi' \land \seq{t}=\seq{t}' \Rightarrow h \right),
\]
where $\set{\seq{x}'}
= \fvs{\seq{t}} \cup \fvs{A'} \cup \fvs{\phi'} \cup \fvs{h}$.
Therefore, we get
\[
\models \sembrack{\Gamma}(A) \Rightarrow \bigwedge \sigma A' \land \exists \seq{x}.(\sigma \phi') \land \sigma \seq{t}=\seq{t}' \Rightarrow \sigma h,
\]
where $\set{\seq{x}} = \fvs{\phi'} \setminus \dom{\sigma}$.  We thus
obtain
\[
\models \sembrack{\Gamma}(A) \land \bigwedge \sigma A' \land \phi \Rightarrow
(\sigma h \lor \forall \seq{x}. \neg (\sigma \phi')).
\]
\end{itemize}
\end{proof}

\begin{lemma}
\label{lem:unfold1}
if $k_1 \leq k_2$, then
$\models \sem{\Gamma}{A}{k_2} \Rightarrow \sem{\Gamma}{A}{k_1}$ holds.
\end{lemma}
\begin{proof}
By the definition of $\sem{\bullet}{\bullet}{\bullet}$.
\end{proof}

\begin{lemma}
\label{lem:unfold2}
We have
\[
\leastmodeld{\D} \models \sem{\Gamma}{A}{k} \land P(\seq{t})
\Rightarrow \sem{\Gamma}{A \cup \set{\mkk{P}{\emptyset}{\circ}(\seq{t})}}{k}.
\]
\end{lemma}
\begin{proof}
By the definition of $\sem{\bullet}{\bullet}{\bullet}$.
\end{proof}

\begin{lemma}
\label{lem:unfold3}
If $\mkk{P}{M}{\alpha}(\seq{t}) \in A$, then for each $i \in \set{1,\dots,m}$,
\[
\leastmodeld{\D} \models \sem{\Gamma}{A}{k+1} \Rightarrow \bigvee_{i=1}^m \exists \seq{x_i}.\left( \phi_i \land \sem{\Gamma}{A \cup \mkk{A_i}{M \cup \set{\alpha}}{\circ}}{k} \right),
\]
where $\D(P)(\seq{t}) = \bigvee_{i=1}^m \exists \seq{x_i}.\left( \phi_i \land \bigwedge A_i \right)$.
\end{lemma}
\begin{proof}
By the definition of $\sem{\bullet}{\bullet}{\bullet}$.
\end{proof}

\begin{lemma}
\label{lem:gamma}
Suppose that $\mkk{P}{M}{\circ}(\seq{t}) \in A$ and $\alpha$ does not
occur in $\Gamma$ or $A$.  We then obtain
\begin{align*}
\models
& \sem{\Gamma}{A}{k} \land \sem{\set{(\alpha \triangleright P(\seq{t}),A,\phi,h)}}{\set{\mkk{P}{M}{\alpha}(\seq{t})}}{k}
\Rightarrow \\
& \sem{\Gamma \cup \set{(\alpha \triangleright P(\seq{t}),A,\phi,h)}}
{(A \setminus \mkk{P}{M}{\circ}(\seq{t})) \cup \set{\mkk{P}{M}{\alpha}(\seq{t})}}{k}
\end{align*}
\end{lemma}
\begin{proof}
By the definition of $\sem{\bullet}{\bullet}{\bullet}$,
we get
\begin{align*}
\models
& \sem{\Gamma}{A}{k} \land \sem{\set{(\alpha \triangleright P(\seq{t}),A,\phi,h)}}{A}{k} \land \\
& \sem{\Gamma}{\set{\mkk{P}{M}{\alpha}(\seq{t})}}{k} \land
  \sem{\set{(\alpha \triangleright P(\seq{t}),A,\phi,h)}}
      {\set{\mkk{P}{M}{\alpha}(\seq{t})}}{k}
\Rightarrow \\
& \sem{\Gamma \cup \set{(\alpha \triangleright P(\seq{t}),A,\phi,h)}}
{(A \setminus \mkk{P}{M}{\circ}(\seq{t})) \cup \set{\mkk{P}{M}{\alpha}(\seq{t})}}{k}
\end{align*}
Because $\alpha$ does not occur in $\Gamma$ or $A$ and
$\mkk{P}{M}{\circ}(\seq{t}) \in A$, we obtain
\begin{align*}
\models
& \sem{\Gamma}{A}{k} \land 
  \sem{\set{(\alpha \triangleright P(\seq{t}),A,\phi,h)}}
      {\set{\mkk{P}{M}{\alpha}(\seq{t})}}{k}
\Rightarrow \\
& \sem{\Gamma \cup \set{(\alpha \triangleright P(\seq{t}),A,\phi,h)}}
{(A \setminus \mkk{P}{M}{\circ}(\seq{t})) \cup \set{\mkk{P}{M}{\alpha}(\seq{t})}}{k}
\end{align*}
\end{proof}

Because $\leastmodeld{\D}$ is the least interpretation, we obtain:
\begin{lemma}
\label{lem:min}
For all $P \in \pvs{\D}$ and $\psi$ with
$\fvs{\psi} \subseteq \set{\seq{x}}$, we get
\[
\leastmodeld{\D} \models
\forall \seq{x}.
\left( \set{P \mapsto \lambda \seq{x}.\psi}\D(P)(\seq{x}) \Rightarrow \psi \right)
\Rightarrow \forall \seq{x}.(P(\seq{x}) \Rightarrow \psi)
\]
\end{lemma}
\begin{lemma}
\label{lem:mincomp}
For all $P \in \pvs{\D}$ and $\varphi$ with
$\fvs{\varphi} \subseteq \set{\seq{x}}$, we get
\[
\leastmodeld{\D} \models
\forall \seq{x}.
\left( \left( \forall p \prec P. \set{P \mapsto \lambda \seq{x}.\varphi}p(\seq{x}) \right) \Rightarrow \varphi \right)
\Rightarrow \forall \seq{x}.(P(\seq{x}) \Rightarrow \varphi),
\]
where $p \prec P$ is a predicate obtained by unfolding $P$ at least
once.
\end{lemma}
\begin{proof}
By Lemma~\ref{lem:min} with $\psi = \forall p \preceq
P. \set{P \mapsto \lambda \seq{x}.\varphi}p(\seq{x})$.
\end{proof}
\begin{lemma}
\label{lem:min2}
Suppose that $P \in \pvs{\D}$, $\mkk{P}{M}{\circ}(\seq{t}) \in A$,
$\alpha$ does not occur in $A$, and the following holds
\[
\leastmodeld{\D} \models
\sem{\set{(\alpha \triangleright
P(\seq{t}),A,\phi,h)}}{\set{\mkk{P}{M}{\alpha}(\seq{t})}}{k}
\land \bigwedge A \land \phi \Rightarrow h
\]
It then follows that
\(\leastmodeld{\D} \models \bigwedge A \land \phi \Rightarrow h\).
\end{lemma}
\begin{proof}
By Lemma~\ref{lem:mincomp}.
\end{proof}